\UseRawInputEncoding

\documentclass[3p,10pt]{elsarticle}
\usepackage{amssymb,amsmath,amsthm,amscd,latexsym}
\usepackage{mathrsfs}
\usepackage{mathrsfs}
\usepackage{amsfonts}
\usepackage{amsmath}
\usepackage{amssymb}
\usepackage{amscd}
\usepackage[all]{xy}
\usepackage{xypic}
\usepackage{graphicx}
\usepackage{epstopdf}
\usepackage{yfonts}
\usepackage[T1]{fontenc}
\usepackage[all]{xy}
\usepackage{times}

\usepackage{caption}
\usepackage[linesnumbered,ruled,vlined]{algorithm2e}

\newtheorem{Definition}{Definition}
\newtheorem{Theorem}{Theorem}

\newtheorem{Example}{Example}[section]
\newtheorem{Remark}{Remark}

\newtheorem{Definition and Notation}{Definition and Notation}
\newtheorem{Lemma}{Lemma}

\newtheorem{Proposition}{Proposition}
\newtheorem{Corollary}{Corollary}
\newenvironment{Proof}[1][Proof]{\noindent {\bf Proof.\;}}{\qed}

\journal{}

\usepackage{amssymb}

\usepackage[figuresright]{rotating}

\begin{document}

\begin{frontmatter}

\title{Galois hulls of cyclic serial codes over a finite chain ring}

\author[USTHB]{Sarra Talbi}\ead{talbissarra@gmail.com}
\author[USTHB]{Aicha Batoul}\ead{a.batoul@hotmail.fr}
\author[UN]{Alexandre Fotue Tabue}\ead{alexfotue@gmail.com}
\author[UVa]{Edgar Mart\'inez-Moro\tnoteref{t1}}\ead{edgar.martinez@uva.es}

\tnotetext[t1]{This author is partially funded by the Spanish
Research Agency (AEI) under Grant PGC2018-096446-B-C21. }

\address[USTHB]{Faculty of Mathematics USTHB, University of Science and Technology of Algiers, Algeria}
\address[UN]{Department of Mathematics, HTTC Bertoua, The University of Ngaoundéré, Cameroon}
\address[UVa]{Institute of Mathematics, University of Valladolid, Castilla, Spain}

\begin{abstract}
In this paper we explore  some properties of Galois hulls of
cyclic serial codes over a chain ring  and  we devise an algorithm
for computing all the possible parameters of the Euclidean hulls
of that codes. We also establish
 the average $p^r$-dimension of the Euclidean
hull, where $\mathbb{F}_{p^r}$ is the residue field of $R$, and we
provide some results of its relative growth.
\end{abstract}


\end{frontmatter}

\section{Introduction}

The Euclidean hull of a linear code is defined to be the
intersection of a code and its Euclidean dual. It was
originally introduced by Assume and Key\cite{AK90} to classify
finite projective planes. Knowing the hull of a linear code is a
key point to determine the complexity of some algorithms for
investigating permutation of two linear codes and computing the
automorphism group of the code, see \cite{Leo82, Pet97, Sen00}. In
general, those algorithms have been proved to be very effective if
the size of the Euclidean hull is small. In the case of codes over
finite fields, Sendrier \cite{Sen97} established the number of
linear codes of length $n$ with a fix dimension Euclidean hull,
also Skersys\cite{Ske03} discussed the average dimension of the
Euclidean hull of cyclic codes. Later, Sangwisut et al.
\cite{SJLU15} determined the dimension of the Euclidean hull of
cyclic and negacyclic codes of length $n$ over a finite field.
Furthermore, Jitman and Sangwisut \cite{JS18} gave the average
Euclidean hull dimension of negacyclic codes over a finite field.
Recently, the concept of the Euclidean hulls has been generalized
to cyclic codes of odd length over $\mathbb{Z}_4$ by Jitman et al.
\cite{JSU19} where the authors provided  an algorithm to determine
the type of the Euclidean hull of cyclic codes over
$\mathbb{Z}_4$.

An important class of linear codes over rings is the class of
cyclic codes.
 They have been studied in a
series of papers ( see \cite{DL05, FM19, AGG, HL00, Rua, NS00}). In
particular, Dinho and Permouth \cite{DL05} gave algebraic
structure of simple root cyclic codes over finite chain rings $R$.
Mart\'inez and R\'{u}a \cite{Rua} generalized these results to
multivariable cyclic codes over $R$. Free cyclic serial codes have
been determined by using cyclotomic cosets and trace map over
finite chain rings $R$ by Fotue and Mouaha in \cite {FM19}. It is
clear that the Euclidean hull of cyclic codes is also cyclic, two
special families of cyclic codes are of great interest, namely
linear complementary dual codes, which are codes whose Euclidean
hull is trivial (see for example \cite{BFMBB20}) and
self-orthogonal codes, which are linear codes whose Euclidean
hulls is the whole code (see for example \cite{SKS18}. These works
motivate us to study the Galois hulls of cyclic codes over finite
chain rings. In this paper, we focus on the study of the Galois
hulls of cyclic codes of length $n$ over a finite chain ring $R$
of parameters $(p, r, a, e, r)$ such that $n$ and $p$ are coprime.
This is the serial case stated in \cite{Rua}, i.e. the cyclic
codes over $R$ whose length $n$ is coprime with $p$ are serial
modules over $R$.  We will generalize the techniques used in
\cite{JSU19}( for $\mathbb{Z}_4$ ) to obtain the parameters and
the average $p^r$-dimensions Euclidean hull of cyclic serial codes
over finite chain rings.

The paper is organized as follows. In Section~\ref{sec:prel}, some
preliminary concepts and some basic results are recalled. In
Section 3, we characterize Galois hulls of cyclic serial code over
finite chain rings. Section 4  shows  the parameters and the
$p^r$-dimensions of the Euclidean hull of cyclic serial codes.
Finally, the average dimension of the Euclidean hull of cyclic
serial codes is computed in Section 5.

\section{Preliminaries}\label{sec:prel}

\subsection{Chain rings}
For an account on the results on finite rings in this section
check \cite{M74}. Throughout  this paper, $p$ is a prime number,
$a, e, r$ are positive integers and $\mathbb{Z}_{p^a}$ is the
residue ring of integers modulo $p^a$.  $R$ will denote a finite
commutative chain ring of characteristic $p^a$ and we will denote
its maximal ideal by $\texttt{J}(R)$ and   $R^\times$ will denote
its multiplicative group. Note that since $R$ is a chain ring  it
is a principal ideal ring, we will denote as $\theta\in R$  a
generator of $\texttt{J}(R)$ and  the ideals of $R$ form a chain
under inclusion $\{ 0\} = \texttt{J}(R)^s\subsetneq
\texttt{J}(R)^{s-1}\subsetneq \cdots\subsetneq
\texttt{J}(R)\subsetneq R$ and $\texttt{J}(R)= \theta^t R$ for $0
\leq t < s$.

The ring epimorphism $\pi : R\rightarrow
 R/\texttt{J}(R)\simeq\mathbb{F}_{p^r}$
naturally extends a ring epimorphism of polynomial rings from
$R[X]$  to $\mathbb{F}_{p^r}[X]$ and on the other hand it
naturally induces an $R$-module epimorphism from $R^n$ to
$(\mathbb{F}_{p^r})^n$, we will abuse the notation and  we will
denote both mappings by $\pi$.

A polynomial $f$ is \emph{basic-irreducible} over $R$ if $\pi(f)$
is irreducible over $\mathbb{F}_{p^r}$. We will denote by denoted
$\texttt{GR}(p^a,r)$ the \emph{Galois ring} of characteristic
$p^a$ and cardinality $p^{ra}$, which is  the quotient ring
$\mathbb{Z}_{p^a}[X]/\langle\,f\,\rangle$ where
$\langle\,f\,\rangle\in\mathbb{Z}_{p^a}[X]$ is the ideal generated
by a monic basic-irreducible polynomial $f$ of degree $r$ over
$\mathbb{Z}_{p^a}$. The  ring  $\texttt{GR}(p^a,r)$  is uniquely
determined up to ring-isomorphism and  for all positive integers
$r_1$ and $r_2$, the  ring $\texttt{GR}(p^a,r_1)$ is a subring of
$\texttt{GR}(p^a,r_2)$ if and only if $r_1$ divides $r_2$.

From the above setting for a given finite chain ring $R$ there is
a 5-tuple $(p, a, r, e, s)$ of positive integers, the so-called
parameters of $R$, such that $R=\texttt{GR}(p^a,r)[\theta]$, and
$\langle\theta\rangle=\texttt{J}(R)$,  $\theta^e\in
p(\mathbb{Z}_{p^a}[\theta])^\times$ and $\theta^{s-1} \neq
\theta^s = 0_R$. From now on, we will denote as $S_d$ the subring
of $R$ such that $S_d:=\texttt{GR}(p^a, d)[\theta]$ and $d$ is a
divisor of $r$. The \emph{Teichmüller set} of $R$ will be denoted
as $\Gamma(R)$ and it is defined as $\Gamma(R)=\{ 0 \}\cup\{a\in
R\,:\,a^{p^r-1}=1\}$. It is the only cyclic subgroup of $R^\times$
isomorphic to the multiplicative group of $\mathbb{F}_{p^r}$. For
each element $a$ in $R$, there is a unique $ (a_0, a_1, \cdots,
a_{s-1})$ in $\Gamma(R)^s$ such that $a =
a_0+a_1\theta+\cdots+a_{s-1}\theta^{s-1}$.

Let $R$ and $S$ be two finite commutative chain rings such that
$S\subset R$ and $1_R=1_S$, then we say that $R$ is an extension
of $S$. Provided that $J(R)$ and  $J(S)$ are their maximal ideals,
we say that the extension is separable if $J(S)R=J(R)$. The group
$G$ of all automorphism $\gamma$ of $R$ such that $\gamma_{|S}$ is
the identity is known as the Galois group of the extension. A
separable extension is called Galois if $\{r\in R\,:\,
\gamma(r)=r,\forall \gamma\in G\}=S$. This condition is equivalent
to the condition $R$ is ring-isomorphic to $S[X]/\langle
f\rangle$, where $f$ is a monic basic irreducible polynomial in
$S[X]$, see \cite[Section 4]{W92}\cite[Theorem XIV.8]{M74}.

We will denote by  $\texttt{Aut}_{S}(R)$  the Galois group of the
Galois extension $S|R$. Let $d$ be positive divisor of $r$, and
let us consider $S=\mathbb{Z}_{p^a}[\theta]$, $R=
\texttt{GR}(p^a,r)[\theta]$, $S_d=\texttt{GR}(p^a,d)[\theta]$, and
$$\texttt{GSub}(S|R):=\{S_d\;:\; d \text{ is a divisor of }r
\text{ and }\mathbb{Z}_{p^a}[\theta]\subseteq S_d\}.$$ It is well
known that $\texttt{Aut}_{S}(R)$ is a cyclic group generated by
the \emph{Frobenius automorphism} $\sigma: R\rightarrow R$ given
by
$\sigma\left(\sum\limits_{t=0}^{s-1}a_t\theta^t\right)=\sum\limits_{t=0}^{s-1}a_t^p\theta^t$,
and therefore, the set $\texttt{Sub}(\texttt{Aut}_{S}(R))$ of
subgroups of $\texttt{Aut}_{S}(R)$ is given by
$$\texttt{Sub}(\texttt{Aut}_{S}(R))=\{\langle\sigma^d\rangle\;:\;
d \text{ is a divisor of }r\}.$$ 

There is a   Galois
correspondence $(\texttt{Stab}; \texttt{Fix})$ between
$\texttt{GSub}(S|R)$ and $\texttt{Sub}(\texttt{Aut}_{S}(R))$ as
follows $\texttt{Stab}
:\texttt{GSub}(S|R)\rightarrow\texttt{Sub}(\texttt{Aut}_{S}(R))$
and  $\texttt{Fix}
:\texttt{Sub}(\texttt{Aut}_{S}(R))\rightarrow\texttt{GSub}(S|R)$
where $\texttt{Stab}(S_d)=\langle\sigma^d\rangle$ and
$\texttt{Fix}(\langle\sigma^d\rangle)=S_d$, where $d$ is a divisor
of $r$ (see \cite{FM19}). The diagram
\begin{center}
\includegraphics[width=6cm]{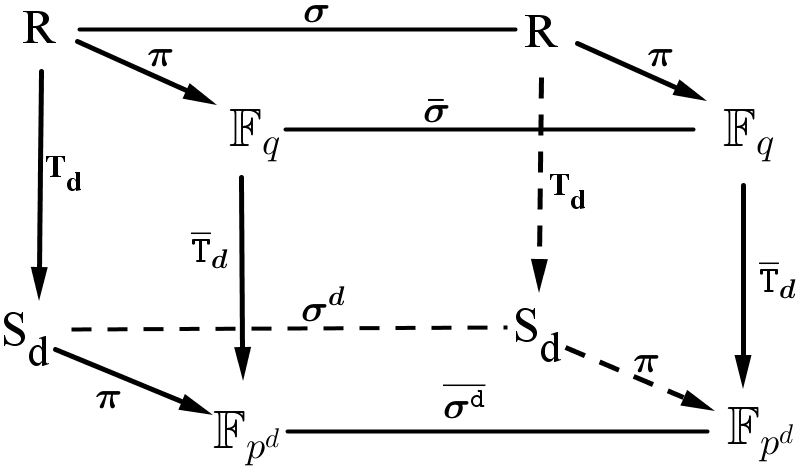}
\end{center}
commutes, where $\overline{\sigma}$ denotes a generator
$\texttt{Aut}_{\mathbb{F}_{p}}(\mathbb{F}_{p^r})$,
$\texttt{T}_d=\sum\limits_{i=0}^{\frac{r}{d}-1}\sigma^{id}$ and
$\overline{\texttt{T}}_d=\sum\limits_{i=0}^{\frac{r}{d}-1}\overline{\sigma}^{id}$.
For any $\textbf{x} = (x_1, \cdots , x_n) \in R^n$, and any matrix
$\mathrm{G} = (a_{i;j})_{k\times n}$ over $R$, $\sigma$ will act
on them component-wise as follows: $\sigma(\textbf{x}) =
(\sigma(x_1), \cdots , \sigma(x_n))$, $\sigma(\mathrm{G}) =
(\sigma(a_{i;j} ))_{k\times n}$.

\subsection{Codes over chain rings}

A \emph{linear code} $C$ of length $n$ over the ring  $R$  is
defined to be a submodule of the $R$-module $R^n$.  We will denote
by
 $\{\textbf{0}\}$, the zero-submodule where $\textbf{0}=(0,0,\ldots,0)\in R^n$. A
linear code $C$ over $R$ is free if $C\cong R^k$ as $R$-modules
for some positive integer $k$.  The residue code of a linear code
$C$ over $R$  is the linear code $\pi(C)$ over $\mathbb{F}_q$
$$\pi(C)=\left\{(\pi(c_0), \pi(c_1), \cdots , \pi(c_{n-1})\,:\,
(c_0, c_1, \cdots, c_{n-1})\in C\right\}.$$ In \cite{FM19}, the
authors introduced and studied the Galois closure of a linear code
$C$ over $R$ of length $n$ as follows,
$\texttt{Cl}_d(C)=\texttt{Ext}(\texttt{T}_d(C))$, where
$\texttt{Ext}(\texttt{T}_d(C))$ is the linear code over $R$ of all
$R$-combinations of codewords in the linear code $\texttt{T}_d(C)$
over $S_d$. A linear code $C$ over $R$ is
$\langle\sigma^d\rangle$-invariant, if $\sigma^d(C)=C$, where $d$
is a divisor of $r$. Recall that for any linear code $C$ over $R$
of length $n$, its subring subcode is given by $\texttt{Res}_d(C)=
C \cap(S_d)^n$. In \cite{FM19} it is shown that
$\langle\sigma^d\rangle$-invariant, if and only if,
$\texttt{T}_d(C)=\texttt{Res}_d(C)$ if and only if,
$C=\texttt{Ext}(\texttt{Res}_d(C))$. For $\ell\in\{0, 1, \ldots ,
r-1\}$ we equip $R^n$ with the $\ell$-Galois inner-product defined
as follows:
$$\langle\,\textbf{u},\,
\textbf{v}\rangle_\ell=\sum\limits_{j=0}^{n-1}u_j\sigma^\ell(v_j),
\quad \hbox{ for all } \textbf{u},\textbf{v}\in  R^n.$$ When
$\ell=0$ it is just the usual Euclidean inner-product and if $r$
is even and $r = 2\ell$ it is the Hermitian inner-product. The
$\ell$-Galois dual of a linear code $C$ over $R$ of length $n$,
denoted $C^{\perp_\ell}$, is defined to be the linear code
$$C^{\perp_\ell}=\left\{\textbf{u}\in R^n\;:\;
 \langle\,\textbf{u},\, \textbf{c}\rangle_\ell=0_R\text{
for all }\textbf{c}\in C\right\}.$$ If $C \subseteq
C^{\perp_\ell}$, then $C$ is said to $\ell$-Galois
\emph{self-orthogonal}. Moreover, $C$ is to be $\ell$-Galois
self-dual if $C = C^{\perp_\ell}$.  The two statements in
Proposition \ref{sidc} below follow immediately from the identity
$$\langle\,\textbf{u},\,
\textbf{v}\rangle_\ell=\langle\,\textbf{u},\,
\sigma^h(\textbf{v})\rangle_{\ell-h}=\sigma^h\left(\langle\,\sigma^{\ell-h}(\textbf{v}),\,
\textbf{u}\rangle_{r-h}\right),$$  for all $0\leq h\leq\ell$,
$\textbf{u}$ and $\textbf{v}$ in $R^n$ where the action is taken
componentwise $\sigma^\ell(\textbf{v})=(\sigma^\ell(v_0),\cdots,
\sigma^\ell(v_{n-1}))$. The following proposition is a generalized
Delsarte's Theorem.

\begin{Proposition}\label{fm19}(\cite[Theorem 3.3]{FM19})
Let $C$ be a linear code over $R$ of length $n$. Then for any
$\ell\in\{0, 1,\cdots, r-1\}$,
$\texttt{T}_d(C^{\perp_\ell})=\left(\texttt{Res}_d(C)\right)^{\perp_\ell}$.
\end{Proposition}

 The following proposition is a natural
generalization to finite chain rings of \cite[Proposition
2.2]{HX20}.

\begin{Proposition}\label{sidc}
Let $C$ be a linear code over $R$ of length $n$. Then
\begin{enumerate}
    \item $\left(\sigma^h(C)\right)^{\perp_{\ell}}=\sigma^h(C^{\perp_{\ell}})$, and $C^{\perp_\ell}=\sigma^h(C^{\perp_{\ell-h}})$,
for any $0\leq h\leq\ell$;
    \item $(C^{\perp_\ell})^{\perp_{h}}=\sigma^{2r-\ell-h}(C)$, for all $0\leq \ell, h\leq r-1$.
\end{enumerate}
\end{Proposition}

From Proposition \ref{sidc} and \cite[Theorem 3.1]{HL00}, we
obtain the following result.

\begin{Corollary}\label{sum-inter}
Let $C$ and $C'$ be linear codes  over $R$ of length $n$. Then
\begin{enumerate}
    \item $(C+C')^{\perp_\ell}=C^{\perp_\ell}\cap C'^{\perp_\ell}$;
    \item $(C\cap C')^{\perp_\ell}=C^{\perp_\ell}+ C'^{\perp_\ell}$.
\end{enumerate}
\end{Corollary}

\begin{Definition}\label{hull} Let $C$ be a linear code over $R$. The $\ell$-\emph{Galois hull} of $C$
will be denoted as $\mathcal{H}_\ell(C)$ and is the intersection
of $C$ and its $\ell$-Galois dual $$\mathcal{H}_{\ell}(C)=C\cap
C^{\perp_\ell}.$$
\end{Definition}

A linear code $C$ over $R$ is said to $\ell$-\emph{Galois LCD} if
$\mathcal{H}_\ell(C)=\{\textbf{0}\}$, and $C$ is said be
$\ell$-\emph{Galois self-orthogonal} if $\mathcal{H}_\ell(C)=C$.
If  we denote that for all $0\leq \ell; h\leq r-1$, we have
$\sigma^h(\mathcal{H}_\ell(C))=\mathcal{H}_\ell(\sigma^h(C))$, and
$\mathcal{H}_\ell(C)=\mathcal{H}_{r-\ell}(C^{\perp_{\ell}})$. From
generalized Delsarte's Theorem in Proposition~\ref{fm19} it
follows that
$\texttt{T}_d(\mathcal{H}_\ell(C))=\left(\texttt{Res}_d(\mathcal{H}_{r-\ell}(C))\right)^{\perp_\ell}$.
Note that if $C$ is $\langle\,\sigma^\ell\,\rangle$-invariant,
then $\mathcal{H}_\ell(C)=\mathcal{H}_0(C)$.

From \cite[Proposition 3.2 and Theorem 3.5]{NS00}, for any linear
code $C$ over $R$ of length $n$, there is a unique $s$-tuple
$(k_0, k_1, \cdots, k_{s-1})$ of positive integers, such that $C$
has a generator matrix in standard form
\begin{align*} \left(
\begin{array}{ccccccc}
  \mathrm{I}_{k_0} & \mathrm{G}_{0,1} & \mathrm{G}_{0,2} & \cdots & \mathrm{G}_{0,s-2 } & \mathrm{G}_{0,s-1} & \mathrm{G}_{0,s} \\
  \mathbf{0} & \theta\mathrm{I}_{k_1} & \theta\mathrm{G}_{1,2} & \cdots & \theta\mathrm{G}_{1,s-2} & \theta\mathrm{G}_{1,s-1} & \theta\mathrm{G}_{1,s} \\
  \cdots & \cdots & \cdots & \cdots & \cdots & \cdots & \cdots \\
  \mathbf{0}  & \mathbf{0} &\mathbf{0}  & \cdots & \mathbf{0} &  \theta^{s-1}\mathrm{I}_{k_{s-1}} & \theta^{s-1}\mathrm{G}_{s-1,s}
\end{array}
\right)\mathrm{U},
\end{align*}
where $\mathrm{U}$ is a permutation matrix and $\mathbf{0}$ all zero matrices of suitable size. The elements
in the $s$-tuple $(k_0, k_1, \cdots, k_{s-1})$ are called
parameters of $C$ and the rank of $C$ is $k_0+ k_1+ \cdots+
k_{s-1}$. From \cite[Theorem 3.10]{NS00}, the parameters of
$C^{\perp_\ell}$ are $(n-k, k_{s-1},\cdots, k_2, k_1)$, where
$k=\texttt{rank}_{R}(C)$. Note that $C$ is free if and only if
$\texttt{rank}_{R}(C)=k_0$ and $k_1= \cdots =k_{s-1}=0$. The
$q$-dimension of a linear code $C$ over $R$, denoted
$\texttt{dim}_q(C)$, is defined to be $\texttt{log}_q(|C|)$. Thus
the $q$-dimension of a linear code $C$ over $R$ of parameters
$(k_0, k_1, \cdots, k_{s-1})$ is
$\sum\limits_{t=0}^{s-1}(s-t)k_t$. Since $R$ is also a Frobenius
ring, it follows that
$\texttt{dim}_q(C)+\texttt{dim}_q(C^{\perp_\ell})=sn$.

\begin{Proposition}
Let $C$ and $C'$ be two linear codes over $R$ of the same length.
Then
$$\texttt{dim}_q(C+C')=\texttt{dim}_q(C)+\texttt{dim}_q(C')-\texttt{dim}_q(C\cap
C').$$ Moreover
$\texttt{dim}_q(\mathcal{H}_\ell(C))=\texttt{dim}_q(\mathcal{H}_{r-\ell}(C))$.
\end{Proposition}

\begin{proof} We have $C=\langle C\backslash C'\rangle\oplus (C\cap
C')$ and $C+C'=\langle C\backslash C'\rangle\oplus C'$. It follows
that  $|C|=|\langle C\backslash C'\rangle|\times|C\cap C'|$ and
$|C+C'|=|\langle C\backslash C'\rangle|\times|C'|$. Thus
$|C+C'|=\frac{|C|}{|C\cap C'|}\times|C'|$. Therefore
$\log_q(|C+C'|)=\log_q(|C|)-\log_q(|C\cap C'|)+\log_q(|C'|).$ From
the definition of $q$-dimension of a linear code we have that
$\texttt{dim}_q(C+C')=\texttt{dim}_q(C)+\texttt{dim}_q(C')-\texttt{dim}_q(C\cap
C').$  Moreover,
\begin{eqnarray*}
  \texttt{dim}_q(\mathcal{H}_\ell(C)) &=& \texttt{dim}_q((C+C^{\perp_{r-\ell}})^{\perp_\ell}), \text{ from Corollary \ref{sum-inter} }; \\
   &=& sn-\texttt{dim}_q(C+C^{\perp_{r-\ell}}),  \text{ since } \texttt{dim}_q(C+C^{\perp_{r-\ell}})+\texttt{dim}_q((C+C^{\perp_{r-\ell}})^{\perp_\ell})=sn; \\
   &=& sn-\left(\texttt{dim}_q(C)+\texttt{dim}_q(C^{\perp_{r-\ell}})-\texttt{dim}_q(\mathcal{H}_{r-\ell}(C))\right);  \\
   &=& \texttt{dim}_q(\mathcal{H}_{r-\ell}(C)), \text{ since } \texttt{dim}_q(C)+\texttt{dim}_q(C^{\perp_{r-\ell}})=sn.
\end{eqnarray*}
\end{proof}

\begin{Proposition}
Let $C$ be a free code over $R$ of length $n$ and $\ell$ be a
positive integer. Then
\begin{enumerate}
    \item $\texttt{dim}_q(\sigma^{\ell}(C))=s\times\texttt{rank}(\sigma^{\ell}(C))=s\times\texttt{dim}_q(\pi(\sigma^{\ell}(C)))$;
    \item $\pi(C)^{\perp_\ell}=\pi(C^{\perp_\ell})$;
    \item $\pi(\mathcal{H}_\ell(C))=\mathcal{H}_\ell(\pi(C))$.
\end{enumerate}
\end{Proposition}

\begin{proof} Since $C$ is free, a generator matrix for $\sigma^{\ell}(C)$ is $\left(
\begin{array}{c|c}
  \mathrm{I}_k & \sigma^{\ell}(\mathrm{A}) \\
\end{array}
\right)\mathrm{U}$, where $\mathrm{A}$ is a $k\times (n-k)$-matrix
over $R$ and $\mathrm{U}$ is a permutation matrix. Thus $\left(
\begin{array}{c|c}
  \mathrm{I}_k & \pi(\sigma^{\ell}(\mathrm{A})) \\
\end{array}
\right)\mathrm{U}$ is a generator matrix for $\pi(C)$. It follows
that $|\sigma^{\ell}(C)|=q^{sk}$ and
$\texttt{rank}(\sigma^{\ell}(C))=\texttt{dim}_q(\pi(\sigma^{\ell}(C)))=k$.
This proves Item 1. Now to proves Item 2. The codes
$\pi(C)^{\perp_\ell}$ and $\pi(C^{\perp_\ell})$ have the same
parity matrix, which is $\left(
\begin{array}{c|c}
  \mathrm{I}_k & \pi(\sigma^{\ell}(\mathrm{A})) \\
\end{array}
\right)\mathrm{U}$. Hence
$\pi(C)^{\perp_\ell}=\pi(C^{\perp_\ell})$. Item 3. is a
consequence of the fact that the above diagram commutes,
$\pi(\mathcal{H}_\ell(C))\subseteq\mathcal{H}_\ell(\pi(C))$ and
$\texttt{dim}_q(\pi(\mathcal{H}_\ell(C)))=\texttt{dim}_q(\mathcal{H}_\ell(\pi(C)))$.
\end{proof}

\section{Galois hulls of cyclic serial codes}

Let $\mathbb{N}$ be the set of nonnegative integers and $n$ be a
positive integer such that $\texttt{gcd}(n, q)=1$. Set
$[|a;b|]=\{a,a+1,\cdots,b\}$ where $(a,b)\in\mathbb{N}^2$ such
that $a< b$. Let $\mathrm{A}, \mathrm{B}$ be two subsets in
$[|0;n-1|]$, as usual, the opposite of $\mathrm{A}$ is
$-\mathrm{A}=\{n-z\;:\; z\in\mathrm{A}\}$ and its complementary
 is $ \overline{\mathrm{A}}=\{z\in[|0;n-1|]\;:\;
z\not\in\mathrm{A}\}$. The set $\mathrm{A}$ is symmetric, if
$\mathrm{A}=-\mathrm{A}$ and the pair $\{\mathrm{A},\mathrm{B}\}$
is asymmetric, if $\mathrm{B}=-\mathrm{A}$.  If
$u\in\mathbb{N}\backslash\{0\}$ then
$u\mathrm{A}=\left\{i\in[|0;n-1|]\,:\,(\exists z\in\mathrm{A})(
uz\equiv i(\texttt{mod}\, j)\right\}$. A subset $\mathrm{Z}$ of
$[|0;n-1|]$ is a $q$-closed set modulo $n$, if
$\mathrm{Z}=q\mathrm{Z}$. The smallest $q$-closed set modulo $n$,
contained a subset $\mathrm{Z}$ of $[|0;n-1|]$ is
$\bigcup_{i\in\mathbb{N}}q^i\mathrm{Z}$ and we will denote it by
$\complement_q(\mathrm{Z})$.  The $q$-cyclotomic cosets modulo $n$
are $\complement_q(\{z\})$, where $z\in[|0;n-1|]$.   We will take
$\complement_q(\emptyset)=\emptyset$ by convention.  It is clear
that the $q$-cyclotomic cosets modulo $n$ form a partition of
$[|0;n-1|]$, and any $q$-closure set modulo $n$ is a union of
$q$-cyclotomic cosets modulo $n$. Denote $[|0;n-1|]_q$, a subset
of $[|0;n-1|]$ such that
$[|0;n-1|]=\bigcup_{z\in[|0;n-1|]_q}\complement_q(z)$. Let $j$ be
a divisor of $n$ such that $\texttt{gcd}(j,q)=1$, we will use the
following notation
\begin{itemize}
    \item $\phi(\,.\,)$ is the Euler totient function;
    \item $\texttt{ord}_j(q)$ the multiplicative order of $q$ modulo $j$;
    \item $\omega(n;q)$ the number of $q$-cyclotomic cosets modulo $n$;
    \item $\mathcal{N}_q = \left\{d \in\mathbb{N}\backslash\{0\} \,:\, (\exists i\in\mathbb{N}\backslash\{0\})(d \,\text{ divides }\,q^i+1) \right\}$;
    \item $\Lambda_j$ the set of symmetric $q$-cyclotomic cosets modulo $n$ of size $\texttt{ord}_j(q)$;
    \item $\overline{\Lambda}_j$ the set of asymmetric pairs of $q$-cyclotomic cosets modulo $n$ of size $\texttt{ord}_j(q)$.
\end{itemize}

The following result is straight forward from Hensel's Lemma \cite{M74}  the uniqueness of this basic-irreducible
factorization.

\begin{Lemma} Let $\delta$ be a generator of the cyclic multiplicative subgroup
$\Gamma(\texttt{GR}(p^a,m))\backslash\{0\}$ of
$(\texttt{GR}(p^a,m))^\times$ where $m=\texttt{ord}_n(q)$. The map
\begin{align}\begin{array}{cccc}
  \Omega : & \left\{\complement_q(\mathrm{Z})\,:\, \mathrm{Z}\subseteq[|0;n-1|]_q \right\} & \rightarrow &  \left\{ f\in  \texttt{GR}(p^a,r)[X]\,:\, f \text{ is monic and } f| X^n-1  \right\}   \\
    & \mathrm{A} & \mapsto &
     \prod\limits_{a\in\mathrm{A}}(X-\delta^a)
\end{array}
\end{align} where $\Omega(\emptyset)=1$, is bijective map. Moreover, for any
$z\in [|0;n-1|]$ and for any sets $\mathrm{A}$ and $\mathrm{B}$
form by union of $q$-cyclotomic cosets modulo $n$ we have
\begin{enumerate}
    \item $\Omega\left(\complement_q(\{z\})\right)$ is a monic
basic-irreducible polynomial over $\texttt{GR}(p^a,r)$ of degree
$\left|\complement_q(\{z\})\right|$;
    \item
$\texttt{lcm}\left(\Omega\left(\mathrm{A}\right),
\Omega\left(\mathrm{B}\right)\right)=\Omega\left(\mathrm{A}\cup\mathrm{B}\right)$
and $\texttt{gcd}\left(\Omega\left(\mathrm{A}\right),
\Omega\left(\mathrm{B}\right)\right)=\Omega\left(\mathrm{A}\cap\mathrm{B}\right)$;
    \item if $\mathrm{A}\cap\mathrm{B}=\emptyset$, then
    $\Omega\left(\mathrm{A}\cup\mathrm{B}\right)=\Omega\left(\mathrm{A}\right)\Omega\left(\mathrm{B}\right)$.
\end{enumerate}
\end{Lemma}

\begin{Proposition}{\cite[Subsection 2.2]{SJLU15}} Let $j$ be a divisor of $n$ such that $\texttt{gcd}(j,q)=1$. Then
\[\gamma(j;q)=|\Lambda_j|=\left\{%
\begin{array}{ll}
     \frac{\phi(j)}{\texttt{ord}_j(q)}, & \hbox{if $j\in\mathcal{N}_q$;} \\
     0, & \hbox{otherwise,}
\end{array}
\right. \textrm{ and } \beta(j;q)=|\overline{\Lambda}_j|=\left\{%
\begin{array}{ll}
     \frac{\phi(j)}{2\texttt{ord}_j(q)}, & \hbox{if $j\not\in\mathcal{N}_q$,} \\
     0, & \hbox{otherwise.}
\end{array}%
\right. \] Moreover,
$\omega(n;q)=\underset{\underset{i\in\mathcal{N}_q}{i\,|\,n}}{\sum}\gamma(i;q)+2\underset{\underset{j\not\in\mathcal{N}_q}{j\,|\,n}}{\sum}\beta(j;q)$.
\end{Proposition}

\noindent We will introduce the following notation
\begin{align}\label{imn1}\mathcal{E}_n(q,s)=\mathcal{I}_n(q,s)\times\left(\mathcal{J}_n(q,s)\right)^2,\end{align}
where
$\mathcal{I}_n(q,s)=\underset{\underset{i\in\mathcal{N}_q}{i\,|\,n}}{\prod}\mathcal{E}_s^{\gamma(i;q)}$
and
$\mathcal{J}_n(q,s)=\underset{\underset{j\not\in\mathcal{N}_q}{j\,|\,n}}{\prod}\mathcal{E}_s^{\beta(j;q)}$,
with
\begin{align}\label{imn2}\mathcal{E}_s=\left\{(x^{(0)},
x^{(1)}, \cdots, x^{(s-1)})\in\{0;1\}^s\;:\;
\sum\limits_{a=0}^{s-1}x^{(a)}\in\{0;1\} \right\}.\end{align} The
elements in $\mathcal{I}_n(q,s)$ are arrays of the form
$(((u_{il}^{(a)})_{0\leq a< s})^{\circ}))$ where
$(u_{il}^{(a)})_{0\leq a< s}$ are in $\mathcal{E}_s$ and the
indices
 $i$ and $j$ satisfy $i\,|\,n, i\in\mathcal{N}_q$ and $1\leq l\leq\gamma(i;q)$,
i.e.,
$$(((u_{il}^{(a)})_{0\leq a< s})^{\circ}))=\left(\left((u^{(a)}_{il})_{0\leq a<s}\right)_{1\leq l\leq\gamma(i;q)}\right)_{i\,|\,n, i\in\mathcal{N}_q}\in\mathcal{I}_n(q,s).$$
Similarly, $(((v^{(a)}_{jh})_{0\leq a< s})^\bullet)=\left(\left(
(v^{(a)}_{jh})_{0\leq a< s})\right)_{1\leq
h\leq\beta(j;q)}\right)_{j\,|\,n,
j\not\in\mathcal{N}_q}\in\mathcal{J}_n(q,s)$. Note that if $s=1,$
then $\mathcal{E}_1=\{0; 1\}$, and in this case, we write
$((u_{il})^{\circ})=(((u_{il}^{(a)})_{0\leq a< 1})^{\circ}))$ and
$((v_{jh})^\bullet)=(((v^{(a)}_{jh})_{0\leq a< 1})^\bullet).$

Let $i$ and $j$ be positive integers such that $i\,|\,n,
i\in\mathcal{N}_q$, and $j\,|\,n, j\not\in\mathcal{N}_q$. In the
sequel, $$\Lambda_i=\{ G_{il}\,:\, 1\leq l \leq
\gamma(i;q)\}\qquad\hbox{ and }\qquad \overline{\Lambda_j}=\{(
F_{jh}, -F_{jh})\,:\, 1\leq h \leq \beta(j;q)\}.$$ Of course,  all
the polynomials in $\{ \Omega(G_{il})\,:\, 1\leq l \leq
\gamma(i;q)\}$ are basic-irreducible in $R[X]$ of degree
$\texttt{ord}_i(q)$, and all the elements in $\{\{\Omega(F_{jh}),
\Omega(-F_{jh})\}\,:\, 1\leq h \leq \beta(j;q)\}$ are pairs of
monic basic-irreducible reciprocal polynomials in $R[X]$ of the
same degree $\texttt{ord}_j(q)$. The basic-irreducible
factorization of $X^n - 1$ in $R[X]$ is given as
\begin{eqnarray}\label{xn}
  X^n-1 &=&
  \underset{\underset{i\in\mathcal{N}_q}{i\,|\,n}}{\prod}\left(\prod\limits_{l=1}^{\gamma(i;q)}\Omega\left(G_{il}\right)\right)\underset{\underset{j\not\in\mathcal{N}_q}{j\,|\,n}}{\prod}\left(\prod\limits_{h=1}^{\beta(j;q)}\Omega\left(F_{jh}\right)\Omega\left(-F_{jh}\right)\right).
\end{eqnarray} Thus, for any monic polynomial $f$ over $R$ dividing $X^n-1$,
there is uniquely $(((u_{il})^\circ), ((v_{jh})^\bullet),
((w_{jh})^\bullet)))$ in $\mathcal{E}_n(q,1)$ such that
\begin{align}\label{d0}
f=
\underset{\underset{i\in\mathcal{N}_q}{i\,|\,n}}{\prod}\left(\prod\limits_{l=1}^{\gamma(i;q)}\Omega\left(G_{il}\right)^{u_{il}}\right)\underset{\underset{j\not\in\mathcal{N}_q}{j\,|\,n}}{\prod}\left(\prod\limits_{h=1}^{\beta(j;q)}\Omega\left(F_{jh}\right)^{v_{jh}}\Omega\left(-F_{jh}\right)^{w_{jh}}\right),\end{align}
and inversely. Denote the right-hand of Eq. (\ref{d0}) by $
\partial(((u_{il})^\circ), ((v_{jh})^\bullet), ((w_{jh})^\bullet))$.
Note that $\partial(((1)^\circ), ((1)^\bullet),
((1)^\bullet))=X^n-1, \partial(((0)^\circ), ((0)^\bullet),
((0)^\bullet))=1$, for all $f_1=\partial(((u_{il})^\circ),
((v_{jh})^\bullet), ((w_{jh})^\bullet))$ and
$f_2=\partial(((u'_{il})^\circ), ((v'_{jh})^\bullet),
((w'_{jh})^\bullet))$, one notes that
\begin{eqnarray*}
  \texttt{lcm}(f_1; f_2) &=& \partial(((\max\{u_{il}, u'_{il}\})^{\circ}),((\max\{v_{jh}, v'_{jh}\})^{\bullet}),((\max\{w_{jh}, w'_{jh}\})^{\bullet})); \\
  \texttt{gcd}(f_1; f_2)&=&\partial(((\min\{u_{il}, u'_{il}\})^{\circ}), ((\min\{v_{jh}, v'_{jh}\})^{\bullet}), ((\min\{w_{jh}, w'_{jh}\})^{\bullet})),
\end{eqnarray*}
and if all $(u_{il}+u'_{il}, v_{jh}+ v'_{jh}, w_{jh}+ w'_{jh})$s
are in $\{0; 1\}^3$ then
\begin{eqnarray*} f_1f_2&=&\partial(((u_{il}+u'_{il})^{\circ}), ((v_{jh}+
v'_{jh})^{\bullet}), ((w_{jh}+
w'_{jh})^{\bullet})).\end{eqnarray*}

A \emph{cyclic code} $C$ of length $n$ over $R$  is a linear code
that is invariant under the transformation $\tau((c_0, c_1, \cdots
, c_{n-1})) = (c_{n-1}, c_0, \cdots , c_{n-2})$. If we denote by
$\langle X^n-1\rangle$ the ideal of $R[X]$ generated by $X^n-1$,
it is well-known that any cyclic code of length $n$  over $R$ can
be represented as an ideal of the quotient ring $R[X]/\langle
X^n-1\rangle$ via the $R$-module isomorphism $ \overline{\Psi}:
R^n \rightarrow R[X]/\langle X^n-1\rangle$, where
$\overline{\Psi}(\textbf{c})=\Psi(\textbf{c})+\langle
X^n-1\rangle$ and
$$
\begin{array}{cccc}
  \Psi: & R^n & \rightarrow & R[X]   \\
    & \textbf{u}=(u_0, u_1,\cdots, u_{n-1}) & \mapsto &  \textbf{u}(X)=u_0+u_1X+\cdots+ u_{n-1}X^{n-1}
\end{array}
$$ which is an $R$-module homomorphism.
We will slightly abuse notation, identifying vectors in $R^n$ as
polynomials in $R[X]$ of degree less than $n$, and vice versa when
the context is clear. It is well-known that $R[X]/\langle
X^n-1\rangle$ is a principal ideal ring and $C$ is a cyclic code
of length $n$ over $R$ if and only if $\overline{\Psi}(C)$ is an
ideal of $R[X]/\langle X^n-1\rangle$, (see \cite{DL05} and
references therein). Thus, the generator polynomial of a cyclic
code $C$ of $R^n$, is the monic polynomial $f$ in $R[X]$ such that
$\overline{\Psi}(C)=\langle\,f(x)\,\rangle$, where
$\langle\,f(x)\,\rangle$ is the ideal of $R[X]/\langle
X^n-1\rangle$ generated by $f$. For a polynomial $f$ of degree $k$
its reciprocal polynomial $X^kf(X^{-1})$ will be denoted by $f^*$
and if $f$ is a factor of $X^n-1$ we denote
$\widehat{f}=\frac{X^n-1}{f}$.   A polynomial $f$ is
\emph{self-reciprocal} if $f=f^\ast$,  otherwise $f$ and $f^\ast$
are called a \emph{reciprocal polynomial pair}.

For any union $\mathrm{A}$ of $q$-cyclotomic cosets modulo $n$,
$\Omega(\mathrm{A})^*=\Omega(-\mathrm{A})$ and
$\widehat{\Omega(\mathrm{A})}=\Omega(\overline{\mathrm{A}})$. The
$(s+1)$-tuple $(\mathrm{A}_0, \mathrm{A}_1, \cdots, \mathrm{A}_s)$
is called to be an ordered $(q,s)$-partition cyclotomic modulo
$n$, if $\mathrm{A}_0, \mathrm{A}_1, \cdots, \mathrm{A}_s$ are
unions of $q$-cyclotomic cosets modulo $n$ whose
$\{\mathrm{A}_t\,:\, \mathrm{A}_t\neq\emptyset,\; \text{ for }
0\leq t\leq s \}$ forms a partition of $[|0;n-1|]$. Denote by
$\Re_n(q,s)$ the set of ordered $(q,s)$-partition cyclotomic
modulo $n$. Note that
$$\Re_n(q,s)=\left\{\left(\complement_q(\lambda^{-1}(\{0\})),
\complement_q(\lambda^{-1}(\{1\})), \ldots,
\complement_q(\lambda^{-1}(\{s\}))\right)\,:\,   \lambda\in
[|0;s|]^{[|0;n-1|]_q} \right\}.$$ It follows that
$|\Re_n(q,s)|=(s+1)^{\omega(n;q)}$.   Let
$\underline{\mathrm{\textbf{A}}}=(\mathrm{A}_0, \mathrm{A}_1,
\cdots, \mathrm{A}_s)$ in $\Re_n(q,s)$. For a positive integer $u$
we denote by
$u\underline{\mathrm{\textbf{A}}}=(u\mathrm{A}_0,u\mathrm{A}_{1},\cdots,
u\mathrm{A}_{s-1})$.  It is easy to see that
$p^\ell\underline{\mathrm{\textbf{A}}}\in\Re_n(q,s)$ for any
$0\leq\ell<r$. From \cite[Theorems 3.4, 3.5 and 3.8]{DL05}, we
have the following result.

\begin{Lemma}\label{lem0}
For any cyclic serial code $C$ over $R$ of length $n$, there is a
unique $(s+1)$-tuple $(\mathrm{A}_0, \mathrm{A}_1, \cdots,
\mathrm{A}_s)$ in $\Re_n(q,s)$ such that
\begin{equation}\label{eq-c2}
   \overline{\Psi}(C) =
   \bigoplus\limits_{t=0}^{s-1}\theta^t\left\langle\,\Omega(\overline{\mathrm{A}_t})(x)\,\right\rangle=  \left\langle\left\{  \theta^t\prod\limits_{a=t+1}^{s}\Omega(\mathrm{A}_{a})\,:\, 0\leq t\leq s-1\right\}\right\rangle.
\end{equation}
  Moreover,  $\overline{\Psi}(C^{\perp_0})=\bigoplus\limits_{t=0}^{s-1}\theta^t\left\langle\,\Omega(-\overline{\mathrm{A}_{s-t}})\,\right\rangle$.
\end{Lemma}

Let $\mathrm{A}$ be a union of $q$-cyclotomic cosets modulo $n$.
From now on, we will consider the code
\begin{align}\label{cc}\mathcal{C}\left(\mathrm{A}\right)=\left\{\textbf{c}\in
R^n\,:\, \Omega(\overline{\mathrm{A}}) \text{ divides
}\Psi(\textbf{c})\right\},\end{align} thus it is clear that
$\overline{\Psi}(\mathcal{C}\left(\mathrm{A}\right))=\left\langle\,\Omega(\overline{\mathrm{A}})(x)\,\right\rangle$.

\begin{Remark}\label{remfc} Free cyclic serial codes over a finite chain ring have been
    studied in \cite{FM19}  using the cyclotomic cosets and the
    trace map.  Note that $\mathcal{C}\left([|0;n-1|]\right)=\{\textbf{0}\}$ and
$\mathcal{C}\left(\emptyset\right)=R^n$. From Lemma \ref{lem0},
for any free cyclic serial code $C$ of length $n$ over $R$  there
exists a unique set   $\mathrm{A}$ which is a union of
$q$-cyclotomic cosets modulo $n$ such that
$C=\mathcal{C}\left(\mathrm{A}\right)$. Moreover,
$\mathcal{C}\left(\mathrm{A}\right)^{\perp_0}=\mathcal{C}\left(-\overline{\mathrm{A}}\right)$,
the generator matrix of $\mathcal{C}\left(\mathrm{A}\right)$ is
$\Omega(\overline{\mathrm{A}})(x)$, and
$\texttt{rank}_R(\mathcal{C}\left(\mathrm{A}\right))=|\mathrm{A}|$.
\end{Remark}

\begin{Proposition} If $\mathrm{A}$ and $\mathrm{B}$ are unions of $q$-cyclotomic cosets modulo $n$, then
\begin{enumerate}
    \item $\mathrm{A}\subseteq \mathrm{B}$ if and only $\mathcal{C}(\mathrm{A})\subseteq \mathcal{C}(\mathrm{B})$;
    \item $\mathcal{C}(\mathrm{A}\cap\mathrm{B})=\mathcal{C}(\mathrm{A})\cap\mathcal{C}(\mathrm{B})$, and $\mathcal{C}(\mathrm{A}\cup\mathrm{B})=\mathcal{C}(\mathrm{A})+\mathcal{C}(\mathrm{B})$;
    \item $\sigma^\ell\left(\mathcal{C}\left(\mathrm{A}\right)\right)=\mathcal{C}\left(p^\ell\mathrm{A}\right)$ and $\mathcal{C}\left(\mathrm{A}\right)^{\perp_\ell}=\mathcal{C}\left(-p^\ell\overline{\mathrm{A}}\right)$, for all $0\leq \ell \leq r-1$.
\end{enumerate}
\end{Proposition}

\begin{proof} Item (1) follows from the definition of $\mathcal{C}(\mathrm{A})$ and $\mathcal{C}(\mathrm{B})$
and the fact that $\mathrm{A}\subseteq \mathrm{B}$ if and only
$\Omega(\overline{\mathrm{B}})$ divides
$\Omega(\overline{\mathrm{A}})$. To prove (2), we note that since
$\mathrm{A}\cap\mathrm{B}\subseteq\mathrm{A}\subseteq\mathrm{A}\cup\mathrm{B}$,
and
$\mathrm{A}\cap\mathrm{B}\subseteq\mathrm{B}\subseteq\mathrm{A}\cup\mathrm{B}$,
from item (1), we have $\mathcal{C}(\mathrm{A}\cap
\mathrm{B})\subseteq\mathcal{C}(\mathrm{A})\cap
\mathcal{C}(\mathrm{B})$ and $\mathcal{C}(\mathrm{A})+
\mathcal{C}(\mathrm{B})\subseteq\mathcal{C}(\mathrm{A}\cup\mathrm{B})$.
Conversely, if $\textbf{c}\in\mathcal{C}(\mathrm{A})\cap
\mathcal{C}(\mathrm{B})$ then $\Omega(\overline{\mathrm{A}})$ and
$\Omega(\overline{\mathrm{B}})$ divide $\Psi(\textbf{c})$. Thus
$\texttt{lcm}(\Omega(\overline{\mathrm{A}})\,,\,\Omega(\overline{\mathrm{B}}))$
divides $\Psi(\textbf{c})$. Now,
$\texttt{lcm}(\Omega(\overline{\mathrm{A}})\,,\,\Omega(\overline{\mathrm{B}}))=\Omega(\overline{\mathrm{A}}\cup\overline{\mathrm{B}})=\Omega(\overline{\mathrm{A}\cap\mathrm{B}})$,
so we have $\mathcal{C}(\mathrm{A})\cap
\mathcal{C}(\mathrm{B})\subseteq\mathcal{C}(\mathrm{A}\cap\mathrm{B})$.
Since
$\texttt{gcd}(\Omega(\overline{\mathrm{A}})\,,\,\Omega(\overline{\mathrm{B}}))=\Omega(\overline{\mathrm{A}}\cap\overline{\mathrm{B}})=\Omega(\overline{\mathrm{A}\cup\mathrm{B}})$,
hence $\mathcal{C}(\mathrm{A})+ \mathcal{C}(\mathrm{B})\supseteq
\mathcal{C}(\mathrm{A}\cup \mathrm{B})$. To finish with the proof
of the item (3), we have
$\sigma^\ell\left(\mathcal{C}\left(\mathrm{A}\right)\right)=\left\{\textbf{c}\in
R^n\;:\;\sigma^\ell\left(\Omega(\overline{\mathrm{A}})\right)
\text{ divides }\Psi(\textbf{c})\right\}$, thus
$\sigma^\ell\left(\mathcal{C}\left(\mathrm{A}\right)\right)=\mathcal{C}\left(p^\ell\mathrm{A}\right)$,
since
$\sigma^\ell\left(\Omega(\overline{\mathrm{A}})\right)=\Omega(p^{\ell}\overline{\mathrm{A}})$.
Finally,  for any $0\leq \ell \leq r-1$ we have
\begin{eqnarray*}
  \mathcal{C}\left(\mathrm{A}\right)^{\perp_\ell} &=& \left(\sigma^\ell\left(\mathcal{C}\left(\mathrm{A}\right)\right)\right)^{\perp_0},\; \text{ from Proposition \ref{sidc}}; \\
    &=& \left(\mathcal{C}\left(p^\ell\mathrm{A}\right)\right)^{\perp_0}; \\
    &=& \mathcal{C}\left(-p^\ell\overline{\mathrm{A}}\right),\;  \text{ from Remark \ref{remfc}}.
\end{eqnarray*}
\end{proof}
Let $\underline{\mathrm{\textbf{A}}}=(\mathrm{A}_0, \mathrm{A}_1,
\ldots, \mathrm{A}_s)$ and
$\underline{\mathrm{\textbf{B}}}=(\mathrm{B}_0, \mathrm{B}_1,
\ldots, \mathrm{B}_s)$ be elements in $\Re_n(q,s)$. We will define
the following set in $R^n$
$$\mathbb{C}(\underline{\mathrm{\textbf{A}}})=\bigoplus\limits_{t=0}^{s-1}\theta^t\mathcal{C}(\mathrm{A}_t).$$
Clearly, $\mathbb{C}(\underline{\mathrm{\textbf{A}}})$ is a cyclic
serial code of length $n$  over $R$ and
$\underline{\mathrm{\textbf{A}}}$ is called the defining multiset
of $\mathbb{C}(\underline{\mathrm{\textbf{A}}})$. The parameters
of $\mathbb{C}(\underline{\mathrm{\textbf{A}}})$  are given by the
entries in $(|\mathrm{A}_0|,|\mathrm{A}_1|, \cdots,
|\mathrm{A}_s|)$  and from Lemma \ref{lem0} it follows that for
any cyclic serial code $C$ over $R$ of length $n$, there is a
unique defining multiset $\underline{\mathrm{\textbf{A}}}$ in
$\Re_n(q,s)$ such that
$C=\mathbb{C}(\underline{\mathrm{\textbf{A}}})$.

Let us denote by
$$\underline{\mathrm{\textbf{A}}}^{\diamond}=(\mathrm{A}_s,\mathrm{A}_{s-1},\ldots,
\mathrm{A}_{0}), \quad
\underline{\mathrm{\textbf{A}}}\sqcup\underline{\mathrm{\textbf{B}}}=(\mathrm{E}_0,
\mathrm{E}_1,\ldots, \mathrm{E}_{s-1})$$ where
$\mathrm{E}_0=\mathrm{A}_0\cup\mathrm{B}_0$, and
$\mathrm{E}_t=(\mathrm{A}_t\cup\mathrm{B}_t)\backslash\left(\bigcup\limits_{i=0}^{t-1}(\mathrm{A}_i\cup\mathrm{B}_i)\right)$
for all $0<t\leq s$. It is easy to see that
$\underline{\mathrm{\textbf{A}}}^{\diamond}$ and
$\underline{\mathrm{\textbf{A}}}\sqcup\underline{\mathrm{\textbf{B}}}$
are in $\Re_n(q,s)$.  Moreover,
$C^{\perp_\ell}=\mathbb{C}(-p^\ell\underline{\mathrm{\textbf{A}}}^{\diamond})$,
and
$\texttt{dim}_q(C)=\sum\limits_{t=0}^{s-1}(s-t)|\mathrm{A}_t|$.
Note that if
$\textbf{\underline{A}}\sqcap\underline{\textbf{B}}=(\underline{\mathrm{\textbf{A}}}^{\diamond}\sqcup\underline{\mathrm{\textbf{B}}}^{\diamond})^{\diamond}=(\mathrm{E}_0,
\mathrm{E}_1, \cdots, \mathrm{E}_s)$, then
$\mathrm{E}_s=\mathrm{A}_s\cup\mathrm{B}_s$ and
$\mathrm{E}_{s-t}=\left(\mathrm{A}_{s-t}\cup\mathrm{B}_{s-t}\right)\backslash\left(\bigcup\limits_{i=0}^{t-1}(\mathrm{A}_{s-i}\cup\mathrm{B}_{s-i})\right)$,
for all $0<t\leq s$. .

\begin{Proposition}{\cite[Theorem 6]{FM19}}\label{interL} Let $\underline{\mathrm{\textbf{A}}}=(\mathrm{A}_0, \mathrm{A}_1, \ldots,
\mathrm{A}_s)$ and $\underline{\mathrm{\textbf{B}}}=(\mathrm{B}_0,
\mathrm{B}_1, \ldots, \mathrm{B}_s)$ in $\Re_n(q,s)$. Then
$\mathbb{C}(\underline{\mathrm{\textbf{A}}})+\mathbb{C}(\underline{\mathrm{\textbf{B}}})
 =\mathbb{C}(\underline{\mathrm{\textbf{A}}}\sqcup\underline{\mathrm{\textbf{B}}})$ and $\mathbb{C}(\underline{\mathrm{\textbf{A}}})\cap\mathbb{C}(\underline{\mathrm{\textbf{B}}})
 =\mathbb{C}(\underline{\mathrm{\textbf{A}}}\sqcap\underline{\mathrm{\textbf{B}}})$.
\end{Proposition}

\begin{Corollary} Let
$\underline{\mathrm{\textbf{A}}}=(\mathrm{A}_0, \mathrm{A}_1,
\ldots, \mathrm{A}_s)$ and
$\underline{\mathrm{\textbf{B}}}=(\mathrm{B}_0, \mathrm{B}_1,
\ldots, \mathrm{B}_s)$ in $\Re_n(q,s)$, and define $g_t
=\prod\limits_{a=t+1}^{s}\Omega(\mathrm{A}_{a})$ and $ h_t
=\prod\limits_{a=t+1}^{s}\Omega(\mathrm{B}_{a})$, for all $0\leq t
< s$. Then
\begin{enumerate}
    \item $\overline{\Psi}\left(\mathbb{C}(\underline{\mathrm{\textbf{A}}})\right)
= \left\langle\{\theta^tg_t(x)\,:\, 0\leq t< s \}\right\rangle$,
and
$\overline{\Psi}\left(\mathbb{C}(\underline{\mathrm{\textbf{B}}})\right)
= \left\langle\{\theta^th_t(x)\,:\, 0\leq t< s \}\right\rangle$;
    \item $\overline{\Psi}\left(\mathbb{C}(\underline{\mathrm{\textbf{A}}}\sqcap\underline{\mathrm{\textbf{B}}})\right)
= \left\langle\{\theta^t\texttt{lcm}(g_t, h_t)(x)\,:\, 0\leq t< s
\}\right\rangle$.
\end{enumerate}
\end{Corollary}

\begin{Proof} We have $\underline{\mathrm{A}}\sqcap\underline{\mathrm{B}}=(\mathrm{E}_0,
\mathrm{E}_1, \ldots, \mathrm{E}_s)$, where
$\mathrm{E}_s=\mathrm{A}_s\cup\mathrm{B}_s$ and
$\mathrm{E}_{s-t}=\left(\mathrm{A}_{s-t}\cup\mathrm{B}_{s-t}\right)\backslash\left(\bigcup\limits_{i=0}^{t-1}(\mathrm{A}_{s-i}\cup\mathrm{B}_{s-i})\right)$,
for all $0<t\leq s$. From Lemma \ref{lem0} it follows that
$\overline{\Psi}\left(\mathbb{C}(\underline{\mathrm{\textbf{A}}})\right)
= \left\langle\{\theta^tg_t(x)\,:\, 0\leq t< s \}\right\rangle$,
and
$\overline{\Psi}\left(\mathbb{C}(\underline{\mathrm{\textbf{B}}})\right)
= \left\langle\{\theta^th_t(x)\,:\, 0\leq t< s \}\right\rangle$.
Since
$\overline{\Psi}\left(\mathbb{C}(\underline{\mathrm{\textbf{A}}}\sqcap\underline{\mathrm{\textbf{B}}})\right)=\overline{\Psi}\left(\mathbb{C}(\underline{\mathrm{\textbf{A}}}\right)\cap\overline{\Psi}\left(\mathbb{C}(\underline{\mathrm{\textbf{B}}})\right)$,
using again  Lemma~\ref{lem0} and Proposition~\ref{interL} it
follows that
$$\overline{\Psi}\left(\mathbb{C}(\underline{\mathrm{\textbf{A}}}\sqcap\underline{\mathrm{\textbf{B}}})\right)=\left\langle
f_0(x), \theta f_1(x),
\ldots,\theta^{s-1}f_{s-1}(x)\right\rangle,$$  where
$f_t=\prod\limits_{a=t+1}^{s}\Omega(\mathrm{E}_{a})$. Thus for all
$0\leq t< s$,
$f_t=\Omega\left(\bigcup\limits_{a=t+1}^{s}\mathrm{E}_{a}\right)$
and
$\bigcup\limits_{a=t+1}^{s}\mathrm{E}_{a}=\bigcup\limits_{a=t+1}^{s}(\mathrm{A}_{s-t-1}\cup\mathrm{B}_{s-t-1})$.
Then
$f_t=\Omega\left(\bigcup\limits_{a=t+1}^{s}(\mathrm{A}_{s-t-1}\cup\mathrm{B}_{s-t-1})\right)=\texttt{lcm}(g_t,
h_t)$.
\end{Proof}

\begin{Theorem}\label{thm1} Let $\underline{\mathrm{\textbf{A}}}$  in $\Re_n(q,s)$. Then
\begin{align}\label{Eq1*}
\mathcal{H}_\ell(\mathbb{C}(\underline{\mathrm{\textbf{A}}}))=\mathbb{C}\left(\underline{\mathrm{\textbf{A}}}\sqcap\,-p^\ell\underline{\mathrm{\textbf{A}}}^{\diamond}\right).
\end{align}
\end{Theorem}

\begin{proof} Let  $\underline{\mathrm{\textbf{A}}}$  in $\Re_n(q,s)$ and
$0\leq\ell<r$. We have
\begin{eqnarray*}
  \mathcal{H}_\ell(\mathbb{C}(\underline{\mathrm{\textbf{A}}})) &=& \mathbb{C}(\underline{\mathrm{\textbf{A}}})\cap\mathbb{C}(\underline{\mathrm{\textbf{A}}})^{\perp_\ell}, \text{ from Definition \ref{hull}  }; \\
    &=& \mathbb{C}(\underline{\mathrm{\textbf{A}}})\cap\mathbb{C}(-p^\ell\underline{\mathrm{\textbf{A}}}^{\diamond}), \text{ since } \mathbb{C}(\underline{\mathrm{\textbf{A}}})^{\perp_\ell}=\mathbb{C}(-p^\ell\underline{\mathrm{\textbf{A}}}^{\diamond});  \\
    &=& \mathbb{C}\left(\underline{\mathrm{\textbf{A}}}\sqcap\,-p^\ell\underline{\mathrm{\textbf{A}}}^{\diamond}\right), \text{ from Proposition \ref{interL}.}
\end{eqnarray*}
\end{proof}

\begin{Example} Let $R=\mathbb{Z}_{2^a}[\theta]$ with $1\leq a \leq 2$ be the finite chain ring of parameters $(2,a,1,e,2)$.
Consider the $2$-cyclotomic cosets modulo $7$ given by
$\complement_2(\{0\})=\{0\}, \complement_2(\{1\})=\{1; 2; 4\}$,
 and $\complement_2(\{3\})=\{3; 5; 6\}$. Note that
$\left\{\complement_2(\{1\}), \complement_2(\{3\})\right\}$ is an
asymmetric set, and $\complement_2(\{0\})$ is a symmetric set.
Consider the cyclic serial code over $R$ of length $7$ with
defining multiset
$\underline{\mathrm{\textbf{A}}}=\left(\complement_2(\{0\}),
\complement_2(\{3\}), \complement_2(\{1\})\right)$.

\noindent Then
$-\underline{\mathrm{\textbf{A}}}^{\diamond}=\left(\complement_2(\{3\}),
\complement_2(\{1\}), \complement_2(\{0\})\right),\text{ and
}\mathbb{C}\left(\underline{\mathrm{\textbf{A}}}\right)=\mathcal{C}\left(\complement_2(\{0\})\right)\oplus
\theta \mathcal{C}\left(\complement_2(\{3\})\right)$. Thus
$\mathbb{C}\left(\underline{\mathrm{\textbf{A}}}\right)^{\perp_0}=\mathbb{C}\left(-\underline{\mathrm{\textbf{A}}}^{\diamond}\right)=\mathcal{C}\left(\complement_2(\{3\})\right)\oplus
\theta \mathcal{C}\left(\complement_2(\{1\})\right)$. Finally,
$\underline{\mathrm{\textbf{A}}}\sqcap
\,-\underline{\mathrm{\textbf{A}}}^{\diamond}=\left(\mathrm{F}_0,
\mathrm{F}_1, \mathrm{F}_2\right)$ where $\mathrm{F}_0=\emptyset,
\mathrm{F}_1=\complement_2(\{3\})$, and
$\mathrm{F}_2=\complement_2(\{0; 1\})$. Therefore
$\mathcal{H}_0(\mathbb{C}\left(\underline{\mathrm{\textbf{A}}}\right))=\mathbb{C}\left(\underline{\mathrm{\textbf{A}}}\sqcap
\,-\underline{\mathrm{\textbf{A}}}^{\diamond}\right)=\mathbb{C}(\emptyset,\complement_2(\{3\}),
\complement_2(\{0; 1\}))= \theta
\mathcal{C}\left(\complement_2(\{0;1\})\right)$.
 \end{Example}

\subsection{Euclidean hulls}
From now on, $\ell=0$. The following result provides us a way of
checking whether a given cyclic serial code $D$ is the Euclidean
hull of a cyclic code $C$ or not. Of course, if
$\mathcal{H}_0(C)=D$, then the cyclic code $D$ is serial if, and
only if $C$ the cyclic code is also a serial code. In the sequel,
for each $\underline{\mathrm{\textbf{X}}}=(\mathrm{X}_0,
\mathrm{X}_1, \cdots, \mathrm{X}_s)\in\Re_n(q,s)$, we will denote
$\Omega(\mathrm{X}_a)=\partial\left(((x^{(a)}_{il})^{\circ}),((y^{(a)}_{jh})^{\bullet}),((z^{(a)}_{jh})^{\bullet})\right)$,
for  $a$ in $\{0; 1; \cdots; s\}$. Thus
$\Omega(-\mathrm{X}_a)=\partial\left(((x^{(a)}_{il})^{\circ}),((z^{(a)}_{jh})^{\bullet}),((y^{(a)}_{jh})^{\bullet})\right)$,
and for $0\leq t\leq s-1$,
$\prod\limits_{a=t+1}^{s}\Omega(\mathrm{X}_{a})=\partial\left(((x_{il}^{[t]})^{\circ}),((y_{jh}^{[t]})^{\bullet}),((y_{jh}^{[t]})^{\bullet})\right)$,
where $x_{il}^{[t]}=\sum\limits_{a=t+1}^{s}x^{(a)}_{il},\,
y_{jh}^{[t]}=\sum\limits_{a=t+1}^{s}y^{(a)}_{jh}$ and
$z_{jh}^{[t]}=\sum\limits_{a=t+1}^{s}z^{(a)}_{jh}$. Note that, for
all $0\leq t<s$, we have
$x_{il}^{[t]}=x_{il}^{[t+1]}+x_{il}^{(t+1)},
y_{jh}^{[t]}=y_{jh}^{[t+1]}+y_{jh}^{(t+1)}$, and
$z_{jh}^{[t]}=z_{jh}^{[t+1]}+z_{jh}^{(t+1)}$. Since
$\partial(((1)^{\circ}),((1)^{\bullet}),((1))^{\bullet})=X^n-1=g_0\partial(((x^{(0)}_{il})^{\circ}),((y^{(0)}_{jh})^{\bullet}),((z^{(0)}_{jh})^{\bullet}))$,
 it follows that
$\sum\limits_{a=0}^{s}x^{(a)}_{il}=\sum\limits_{a=0}^{s}y^{(a)}_{jh}=\sum\limits_{a=0}^{s}z^{(a)}_{jh}=1$.
From Eqs. (\ref{d0}) and (\ref{eq-c2}), there exists a unique
$$(\textbf{x}^{\circ},\textbf{y}^{\bullet},\textbf{z}^{\bullet})=\left(
(((x^{(a)}_{il})_{0\leq a< s})^{\circ}), (((y^{(a)}_{jh})_{0\leq
a< s})^{\bullet}), (((z^{(a)}_{jh})_{0\leq a<
s})^{\bullet})\right)$$ which is an element in
$\mathcal{E}_n(q,s)$ such that
$$\overline{\Psi}(\mathbb{C}(\underline{\mathrm{\textbf{X}}}))=\left\langle\left\{
\theta^t\cdot\partial\left(((x^{[t]}_{il})^{\circ}),((y^{[t]}_{jh})^{\bullet}),((z^{[t]}_{jh})^{\bullet})\right)\,:\,
0\leq t\leq s-1 \right\}\right\rangle.$$

From Eqs. (\ref{xn}), (\ref{d0}), and (\ref{eq-c2}), the following
lemma  follows.

\begin{Lemma}\label{lcs} There is a bijection between the set $\mathcal{C}(n; R)$ of cyclic serial codes of length $n$ over $R$ and
the set $\mathcal{E}_n(q,s)$.
\end{Lemma}

When $\ell=0$, and with the triple-sequence of a cyclic serial
code, in comparing the two sides of Eq.\,(\ref{Eq1*}) of
Theorem\,\ref{thm1}, it obtains the following result.

\begin{Corollary}\label{cor*} Let $(\textbf{x}^{\circ},\textbf{y}^{\bullet},\textbf{z}^{\bullet})=\left(
(((x^{(a)}_{il})_{0\leq a< s})^{\circ}), (((y^{(a)}_{jh})_{0\leq
a< s})^{\bullet}), (((z^{(a)}_{jh})_{0\leq a<
s})^{\bullet})\right)$ and
$$(\textbf{u},\textbf{v},\textbf{w})=\left(
(((u^{(a)}_{il})_{0\leq a< s})^{\circ}), (((v^{(a)}_{jh})_{0\leq
a< s})^{\bullet}), (((w^{(a)}_{jh})_{0\leq
a<s})^{\bullet})\right)$$ in $\mathcal{E}_n(q,s)$  such that
$$\overline{\Psi}(C)=\left\langle\left\{
\theta^t\cdot\partial\left(((x^{[t]}_{il})^{\circ}),((y^{[t]}_{jh})^{\bullet}),((z^{[t]}_{jh})^{\bullet})\right)\,:\,
0\leq t\leq s-1 \right\}\right\rangle,$$ and
$$\overline{\Psi}(D)=\left\langle\left\{
\theta^t\cdot\partial\left(((u^{[t]}_{il})^{\circ}),((v^{[t]}_{jh})^{\bullet}),((w^{[t]}_{jh})^{\bullet})\right)\,:\,
0\leq t\leq s-1 \right\}\right\rangle.$$ Then $\mathcal{H}_0(C)=D$
if, and only if for all $0\leq t\leq s-1$,
\begin{align}
\left\{%
\begin{array}{l}
    u^{[t]}_{il}=\max\left\{\sum\limits_{a=t+1}^{s}x^{(a)}_{il};\sum\limits_{a=t+1}^{s}x^{(s-a)}_{il}\right\}; \\
    v^{[t]}_{jh}=\max\left\{\sum\limits_{a=t+1}^{s}y^{(a)}_{jh};\sum\limits_{a=t+1}^{s}z^{(s-a)}_{jh}\right\}; \\
    w^{[t]}_{jh}=\max\left\{\sum\limits_{a=t+1}^{s}z^{(a)}_{jh};\sum\limits_{a=t+1}^{s}y^{(s-a)}_{jh}\right\}.
\end{array}%
\right.
\end{align}
Moreover, for all $0\leq t\leq s-1$, if $2t\leq s-1$, then
$u^{[t]}_{il}=\max\left\{\sum\limits_{a=t+1}^{s}x^{(a)}_{il};\sum\limits_{a=t+1}^{s}x^{(s-a)}_{il}\right\}=1$,
and $(v^{[t]}_{jh}; w^{[t]}_{jh})\in\{(1; 0), (0; 1), (1; 1)\}$,
since $\sum\limits_{a=0}^{s}x^{(a)}_{il}=1$.
\end{Corollary}

From Corollary \ref{cor*}, we can state that Theorem 3.4, and
Corollaries 3.5 and 3.6  in \cite{JSU19} can be naturally extended
to finite chain rings of nilpotency index 2. The following remark
gives this generalization.

\begin{Remark} Let $(\textbf{x}^{\circ},\textbf{y}^{\bullet},\textbf{z}^{\bullet})=\left(
(((x^{(a)}_{il})_{0\leq a< 2})^{\circ}), (((y^{(a)}_{jh})_{0\leq
a< 2})^{\bullet}), (((z^{(a)}_{jh})_{0\leq a<
2})^{\bullet})\right)$ and
$$(\textbf{u},\textbf{v},\textbf{w})=\left(
(((u^{(a)}_{il})_{0\leq a< 2})^{\circ}), (((v^{(a)}_{jh})_{0\leq
a< 2})^{\bullet}), (((w^{(a)}_{jh})_{0\leq a<
2})^{\bullet})\right)$$ in $\mathcal{E}_n(q,2)$ such that
$$\overline{\Psi}(C)=\left\langle\left\{
\partial\left(((x^{[0]}_{il})^{\circ}),((y^{[0]}_{jh})^{\bullet}),((z^{[0]}_{jh})^{\bullet})\right),
\theta\cdot\partial\left(((x^{[1]}_{il})^{\circ}),((y^{[1]}_{jh})^{\bullet}),((z^{[1]}_{jh})^{\bullet})\right)\right\}\right\rangle,$$
and
$$\overline{\Psi}(D)=\left\langle\left\{
\partial\left(((u^{[0]}_{il})^{\circ}),((v^{[0]}_{jh})^{\bullet}),((w^{[0]}_{jh})^{\bullet})\right),
\theta\cdot\partial\left(((u^{[1]}_{il})^{\circ}),((v^{[1]}_{jh})^{\bullet}),((w^{[1]}_{jh})^{\bullet})\right)\right\}\right\rangle.$$
Then $\mathcal{H}_0(C)=D$ if, and only if $(x_{il}^{(1)};
x_{il}^{(2)})\in\left\{
\begin{array}{ll}
    \{(0; 1)\}, & \hbox{if $u_{il}^{(2)}=0$;} \\
    \{(0; 0), (1; 0)\}, & \hbox{if $u_{il}^{(2)}=1$,}
\end{array}
\right.    $ and
$$(y_{jh}^{(1)}; y_{jh}^{(2)}; z_{jh}^{(1)}; z_{jh}^{(2)})\in\left\{%
\begin{array}{ll}
    \{(0; 0; 0; 0), (1; 0; 1; 0)\}, & \hbox{if $(u_{il}^{[0]}; v_{jh}^{[0]}; w_{jh}^{[0]}; v_{jh}^{[1]}; w_{jh}^{[1]})=(1; 1; 1; 1; 1)$;} \\
    \{(0; 1; 1; 0), (1; 0; 1; 1)\}, & \hbox{if $(u_{il}^{[0]}; v_{jh}^{[0]}; w_{jh}^{[0]}; v_{jh}^{[1]}; w_{jh}^{[1]})=(1; 1; 1; 0; 1)$;} \\
    \{(0; 1; 0; 0), (1; 0; 0; 1)\}, & \hbox{if $(u_{il}^{[0]}; v_{jh}^{[0]}; w_{jh}^{[0]}; v_{jh}^{[1]}; w_{jh}^{[1]})=(1; 1; 1; 1; 0)$;} \\
    \{(1; 0; 0; 0)\}, & \hbox{if $(u_{il}^{[0]}; v_{jh}^{[0]}; w_{jh}^{[0]}; v_{jh}^{[1]}; w_{jh}^{[1]})=(1; 1; 0; 1; 0)$;} \\
    \{(0; 0; 1; 0)\}, & \hbox{if $(u_{il}^{[0]}; v_{jh}^{[0]}; w_{jh}^{[0]}; v_{jh}^{[1]}; w_{jh}^{[1]})=(1; 0; 1; 0; 1)$;} \\
    \{(0; 1; 0; 1)\}, & \hbox{if $(u_{il}^{[0]}; v_{jh}^{[0]}; w_{jh}^{[0]}; v_{jh}^{[1]}; w_{jh}^{[1]})=(1; 1; 1; 0; 0)$,}
\end{array}
\right.$$   for all $i, l, j, h$. Moreover,
\begin{enumerate}
    \item $C$ is  LCD if, and only
    if $x_{il}^{(2)}=x_{il}^{(1)}, y_{jh}^{(2)}=y_{jh}^{(1)}, z_{jh}^{(2)}=z_{jh}^{(1)}$ and $(x_{il}^{(2)}; y_{jh}^{(2)}; z_{jh}^{(2)})\in\{(0; 0; 0), (0; 1; 1), (1; 0; 0), (1; 1; 1)\}$, for all $i, l, j, h$.
    \item $C$ is self-orthogonal if, and only
    if $(x_{il}^{(2)}; x_{il}^{(1)})\in\{(0; 1), (1; 0)\}$ and $$(y_{jh}^{(2)}; y_{jh}^{(1)}; z_{jh}^{(2)}; z_{jh}^{(1)})\in\{(1; 0; 1; 0), (0; 1; 1; 0), (1; 0; 0; 1), (1; 0; 0; 0), (0; 0; 1; 0), (0; 1; 0; 1)\},$$ for all $i, l, j, h$.
\end{enumerate}
\end{Remark}

\section{ The $q$-dimensions of Euclidean hulls of cyclic serial codes}

In this section, $C$ is a cyclic serial code of length $n$ over
$R$ with triple-sequence
$$(\textbf{x}^{\circ},\textbf{y}^{\bullet},\textbf{z}^{\bullet})=\left(
(((x^{(a)}_{il})_{0\leq a< s})^{\circ}), (((y^{(a)}_{jh})_{0\leq
a< s})^{\bullet}), (((z^{(a)}_{jh})_{0\leq a<
s})^{\bullet})\right)$$ in $\mathcal{E}_n(q,s).$ Then
$$\overline{\Psi}(C)=\left\langle\left\{
\theta^t\cdot\partial\left(((x^{[t]}_{il})^{\circ}),((y^{[t]}_{jh})^{\bullet}),((z^{[t]}_{jh})^{\bullet})\right)\,:\,
0\leq t\leq s-1 \right\}\right\rangle.$$ From Corollary
\ref{cor*},
$$\overline{\Psi}\left(\mathcal{H}_0(C)\right)=\left\langle\left\{
\theta^t\cdot\partial\left(((u^{[t]}_{il})^{\circ}),((v^{[t]}_{jh})^{\bullet}),((w^{[t]}_{jh})^{\bullet})\right)\,:\,
0\leq t\leq s-1 \right\}\right\rangle,$$  where
\begin{align*}
\left\{%
\begin{array}{l}
    u^{[t]}_{il}=1-\min\left\{ \sum\limits_{a=0}^{t}x^{(a)}_{il}; 1-\sum\limits_{a=0}^{s-t-1}x^{(a)}_{il}\right\}; \\
    v^{[t]}_{jh}=1-\min\left\{ \sum\limits_{a=0}^{t}y^{(a)}_{jh}; 1-\sum\limits_{a=0}^{s-t-1}z^{(a)}_{jh}\right\}; \\
    w^{[t]}_{jh}=1-\min\left\{ \sum\limits_{a=0}^{t}z^{(a)}_{jh}; 1-\sum\limits_{a=0}^{s-t-1}y^{(a)}_{jh}\right\},
\end{array}%
\right.
\end{align*} for all $0\leq
t\leq s-1$. The following notations are important for the sequel
of this paper.  For all $0\leq t \leq s-1$, $1\leq
l\leq\gamma(i;q)$ and $1\leq h\leq\beta(j;q)$, denote by:
\begin{align}\label{n1}\varepsilon^{(t)}_{jh}=v^{[t]}_{jh}+w^{[t]}_{jh}.\end{align} Note that $\varepsilon^{(-1)}_{jh}=2$.
Let us consider  now \begin{align}\label{n2}
\vartriangle_{il}=\sum\limits_{t=0}^{s-1}(s-t)(u^{[t-1]}_{il}-u^{[t]}_{il}),
\text{ and }
\blacktriangle_{jh}=\sum\limits_{t=0}^{s-1}(s-t)(\varepsilon^{(t-1)}_{il}-\varepsilon^{(t)}_{il}).
\end{align}
Obviously, $\vartriangle_{il} =
\sum\limits_{t=0}^{s-1}\vartriangle_{il}^{(t)}$, where
$\vartriangle_{il}^{(t)}=\min\left\{
\sum\limits_{a=0}^{t}x^{(a)}_{il};
1-\sum\limits_{a=0}^{s-t-1}x^{(a)}_{il}\right\},$ and
$\blacktriangle_{jh}=\sum\limits_{t=0}^{s-1}\blacktriangle_{jh}^{(t)}$,
where
$$\blacktriangle_{jh}^{(t)}=\min\left\{
\sum\limits_{a=0}^{t}y^{(a)}_{jh};
1-\sum\limits_{a=0}^{s-t-1}z^{(a)}_{jh}\right\}+\min\left\{
\sum\limits_{a=0}^{t}z^{(a)}_{jh};
1-\sum\limits_{a=0}^{s-t-1}y^{(a)}_{jh}\right\}.$$  Thus
$\vartriangle_{i}=\sum\limits_{l=1}^{\gamma(i;q)}\vartriangle_{il},$
$\varepsilon^{(t)}_{j}=\sum\limits_{h=1}^{\beta(j;q)}\varepsilon^{(t)}_{jh}\text{
and
}\blacktriangle_{j}=\sum\limits_{h=1}^{\beta(j;q)}\blacktriangle_{jh}.
$ Note that for any integer $t$, if $t<0$ then
$\varepsilon_{jh}^{(t)}=2\beta(j,q)$.

\begin{Remark}\label{r00} Let $0\leq t\leq s-1.$
\begin{enumerate}
    \item  $\vartriangle_{il}^{(t)}\in\{0, 1\}$ and $\blacktriangle_{jh}^{(t)}\in\{0, 1, 2\}.$
    \item If $0<t<s$, then $\vartriangle_{il}^{(t-1)}\leq\vartriangle_{il}^{(t)}$ and $\blacktriangle_{jh}^{(t-1)}\leq\blacktriangle_{jh}^{(t)}.$
    \item If $2t< s$, then $\vartriangle_{il}^{(t)}=0$ and $\blacktriangle_{jh}^{(t)}\leq 1.$
\end{enumerate}
\end{Remark}


\begin{Lemma}\label{lem-co} Let $j$ be a divisor of $n$ such that $j\not\in\mathcal{N}_q$. Then   $$\left\{%
\begin{array}{ll}
    0\leq \varepsilon^{(t-1)}_{j}-\varepsilon^{(t)}_{j}\leq  \beta(j;q)-(\varepsilon^{(t-2)}_{j}-\varepsilon^{(t-1)}_{j}), & \hbox{if $t< \left\lceil\frac{s}{2}\right\rceil$;} \\
    0\leq\varepsilon^{(t-1)}_{j}-\varepsilon^{(t)}_{j}\leq  2\left(\beta(j;q)-(\varepsilon^{(t-2)}_{j}-\varepsilon^{(t-1)}_{j})\right), & \hbox{if  $t\geq \left\lceil\frac{s}{2}\right\rceil$.}
\end{array}%
\right. $$
\end{Lemma}

\begin{Proof} Let $0\leq t \leq s-1$ and $\blacktriangle^{(t)}_{j}=\sum\limits_{h=1}^{\beta(j;q)}\blacktriangle^{(t)}_{jh}$.
We have $\varepsilon^{(t)}_{jh}=2-\blacktriangle_{jh}^{(t)}$. From
Remark \ref{r00}, two cases are considered. Let
$\varpi_j^{(t-1)}=|\{h\in\mathbb{N}\;:\;1\leq h\leq \beta(j;q)
\text{ and } \varepsilon^{(t-1)}_{jh}=1\}|.$
\begin{description}
    \item[Case 1: $t< \left\lceil\frac{s}{2}\right\rceil$.] We have $\varepsilon^{(t)}_{jh}\in\{1, 2\}$, and $\left\{%
\begin{array}{ll}
    \varepsilon^{(t-1)}_{jh}-\varepsilon^{(t)}_{jh}\in\{0, 1\}, & \hbox{if $\varepsilon^{(t-1)}_{jh}=2;$} \\
    \varepsilon^{(t)}_{jh}=\varepsilon^{(t-1)}, & \hbox{if $\varepsilon^{(t-1)}_{jh}=1$.}
\end{array}%
\right.$ Without loss of generality, we assume that
$\varepsilon^{(t-1)}_{jh}=1$, for all $h\in\{1,
\cdots,\varpi_j^{(t-1)}\}.$ Thus
\begin{eqnarray*}
  \varepsilon^{(t-1)}_{j}-\varepsilon^{(t)}_{j} &=& \left(\sum\limits_{h=1}^{\varpi_j^{(t-1)}}(\varepsilon^{(t-1)}_{jh}-\varepsilon^{(t)}_{jh})\right)+\left(\sum\limits_{h=\varpi_j^{(t-1)}+1}^{\beta(j;q)}(\varepsilon^{(t-1)}_{jh}-\varepsilon^{(t)}_{jh})\right); \\
    &=&
    0+\left(\sum\limits_{h=\varpi_j^{(t-1)}+1}^{\beta(j;q)}(\varepsilon^{(t-1)}_{jh}-\varepsilon^{(t)}_{jh})\right),
    \text{ since } 0\leq \varepsilon^{(t-1)}_{jh}-\varepsilon^{(t)}_{jh}\leq 1.
\end{eqnarray*}
Hence $0\leq \varepsilon^{(t-1)}_{j}-\varepsilon^{(t)}_{j}\leq
\beta(j;q)-\varpi_j^{(t-1)}\leq\beta(j;q)-(\varepsilon^{(t-2)}_{j}-\varepsilon^{(t-1)}_{j}).$
    \item[Case 2: $t\geq \left\lceil\frac{s}{2}\right\rceil$.] We have  $\left\{%
\begin{array}{ll}
    \varepsilon^{(t-1)}_{jh}-\varepsilon^{(t)}_{jh}\in\{0, 1, 2\}, & \hbox{if $\varepsilon^{(t-1)}_{jh}\in\{1, 2\};$} \\
    \varepsilon^{(t)}_{jh}=\varepsilon^{(t-1)}_{jh}, & \hbox{if $\varepsilon^{(t-1)}_{jh}=0$.}
\end{array}%
\right.$ Without loss of generality, we assume that
$\varepsilon^{(t-1)}_{jh}=0$, for all $h\in\{1, \cdots,
\varpi^{(t-1)}_{j}\}.$ Thus
\begin{eqnarray*}
  \varepsilon^{(t-1)}_{j}-\varepsilon^{(t)}_{j} &=& \left(\sum\limits_{h=1}^{\varpi^{(t-1)}_{j}}(\varepsilon^{(t-1)}_{jh}-\varepsilon^{(t)}_{jh})\right)+\left(\sum\limits_{h=\varpi^{(t-1)}_{j}+1}^{\beta(j;q)}(\varepsilon^{(t-1)}_{jh}-\varepsilon^{(t)}_{jh})\right); \\
    &=&
    0+\left(\sum\limits_{h=\varpi^{(t-1)}_{j}+1}^{\beta(j;q)}(\varepsilon^{(t-1)}_{jh}-\varepsilon^{(t)}_{jh})\right),
    \text{ since } 0\leq \varepsilon^{(t-1)}_{jh}-\varepsilon^{(t)}_{jh}\leq 2.
\end{eqnarray*}
Therefore $0\leq \varepsilon^{(t-1)}_{j}-\varepsilon^{(t)}_{j}\leq
2(\beta(j;q)-\varpi^{(t-1)}_{j})\leq
2\left(\beta(j;q)-(\varepsilon^{(t-2)}_{j}-\varepsilon^{(t-1)}_{j})\right).$
\end{description}
\end{Proof}

\begin{Theorem}\label{thm2} The parameters of the
Euclidean hull of a cyclic serial code over $R$ of length $n$ are
given by $(k_0, k_1,\cdots, k_{s-1})$ where $2k_0+ k_1+ \cdots+
k_{s-1}\leq n$, \begin{eqnarray*} k_t&=&
  \underset{\underset{i\in\mathcal{N}_q}{i\,|\,n}}{\sum}\texttt{ord}_i(q)\cdot \mu_i^{(t)}+\underset{\underset{j\not\in\mathcal{N}_q}{j\,|\,n}}{\sum}\texttt{ord}_j(q)\cdot \nu_j^{(t)},
\end{eqnarray*}
with $$\left\{%
\begin{array}{ll}
    \mu_i^{(t)}=0, & \hbox{if $t< \left\lceil\frac{s}{2}\right\rceil$;} \\
    0\leq \mu_i^{(t)}\leq \gamma(i; q), & \hbox{if $ t\geq \left\lceil\frac{s}{2}\right\rceil$,}
\end{array}
\right., \text{ and }
\left\{%
\begin{array}{ll}
    \varepsilon_j^{(t)}=0, & \hbox{if $n\in\mathcal{N}_q$;} \\
    0\leq \nu_j^{(t)}\leq  \beta(j;q)-\nu_j^{(t-1)}, & \hbox{if $n\not\in\mathcal{N}_q$, and $ t< \left\lceil\frac{s}{2}\right\rceil$;} \\
    0\leq \nu_j^{(t)}\leq 2(\beta(j;q)-\nu_j^{(t-1)}), & \hbox{if $n\not\in\mathcal{N}_q$, and $ t\geq \left\lceil\frac{s}{2}\right\rceil$.}
\end{array}%
\right. $$ Moreover $\nu_j^{(-1)}=0.$
\end{Theorem}

\begin{Proof} Let $(k_0, k_1,\cdots, k_{s-1})$ be the parameters of
$\mathcal{H}_0(C)$. When $\mathcal{H}_0(C)=C$, we have $2k_0+ k_1+
\cdots+ k_{s-1}\leq n$. Then for all $0\leq t\leq s-1$,
\begin{eqnarray*}
k_t  &=& \texttt{deg}\left(\partial\left(((u^{[t-1]}_{il})^{\circ}),((v^{[t-1]}_{jh})^{\bullet}),((w^{[t]}_{jh})^{\bullet})\right)\right)-\texttt{deg}\left(\partial\left(((u^{[t]}_{il})^{\circ}),((v^{[t]}_{jh})^{\bullet}),((w^{[t]}_{jh})^{\bullet})\right)\right); \\
  &=&
  \underset{\underset{i\in\mathcal{N}_q}{i\,|\,n}}{\sum}\texttt{ord}_i(q)\cdot\mu_i^{(t)}+\underset{\underset{j\not\in\mathcal{N}_q}{j\,|\,n}}{\sum}\texttt{ord}_j(q)\cdot(\varepsilon_j^{(t-1)}-\varepsilon_j^{(t)}),
  \qquad \hbox{ where }
  \mu_i^{(t)}=\sum\limits_{l=1}^{\gamma(i;
    q)}(u_{il}^{[t-1]}-u_{il}^{[t]}).
\end{eqnarray*}

\noindent Since
$\left\{%
\begin{array}{ll}
    \mu_i^{(t)}=0, & \hbox{if $t< \left\lceil\frac{s}{2}\right\rceil$;} \\
    0\leq \mu_i^{(t)}\leq \gamma(i; q), & \hbox{if $ t\geq \left\lceil\frac{s}{2}\right\rceil$,}
\end{array}
\right.$, it follows that $\left\{\begin{array}{ll}
     \mu_i^{(t)}=0, & \hbox{if $2t< s$;} \\
0\leq \mu_i^{(t)}\leq \gamma(i; q), & \hbox{if $ s\leq 2t.$}
\end{array}
\right.$ On the other hand, one notes that if $n\in\mathcal{N}_q$,
then all positive divisor of $n$ is in then $ \mathcal{N}_q.$  By
Lemma \ref{lem-co}, we obtain
$$\left\{%
\begin{array}{ll}
    \varepsilon_j^{(t)}=0, & \hbox{if $n\in\mathcal{N}_q$;} \\
    0\leq \nu_j^{(t)}\leq  \beta(j;q)-\nu_j^{(t-1)}, & \hbox{if $n\not\in\mathcal{N}_q$, and $t< \left\lceil\frac{s}{2}\right\rceil$;} \\
    0\leq \nu_j^{(t)}\leq 2(\beta(j;q)-\nu_j^{(t-1)}), & \hbox{if $n\not\in\mathcal{N}_q$, and $ t\geq \left\lceil\frac{s}{2}\right\rceil$.}
\end{array}%
\right.$$ where
$\nu_j^{(t)}=\varepsilon^{(t-1)}_{j}-\varepsilon^{(t)}_{j}.$
Obviously
$\nu_j^{(-1)}=\varepsilon^{(-2)}_{j}-\varepsilon^{(-1)}_{j}=0.$
\end{Proof}

The previous discussion leads to the following algorithm.
\begin{algorithm}
\DontPrintSemicolon \KwIn{ Length $n$, and a finite chain ring $R$
of parameters $(p, a, r, e, s)$ such that $\texttt{gcd}(p, n)=1$.}
\KwOut{All possible $s$-tuples $(k_0, k_1, \cdots, k_{s-1})$.}
\uIf{$n\in\mathcal{N}_q$}{
      \For{$0\leq t<s$}{\uIf{$t<\left\lceil\frac{s}{2}\right\rceil$}{ $k_t=0$.\; }
    \Else{ For each $i\,|\,n,$ compute $\texttt{ord}_i(q)$, and $\gamma(i;q)$,\;
      therefore all the possible values of $k_t,$ such that $$k_t=\underset{\underset{i\in\mathcal{N}_q}{i\,|\,n}}{\sum}\texttt{ord}_i(q)\cdot \mu_i^{(t)},$$ with $0\leq \mu_i^{(t)}\leq \gamma(i; q)$.
    }}
 \Return{The possible parameters $(0, \cdots, 0,
k_{\left\lceil\frac{s}{2}\right\rceil}, \cdots, k_{s-1})$ such
that $k_{\left\lceil\frac{s}{2}\right\rceil}+ \cdots+ k_{s-1}\leq
n$.}\;
    }
    \Else{For each $i\,|\,n,$ if $i\in\mathcal{N}_q$, then compute $\texttt{ord}_i(q)$, and $\gamma(i;q)$.\;
           For each $j\,|\,n,$ if $j\not\in\mathcal{N}_q$, then compute $\texttt{ord}_j(q)$, and $\beta(j;q)$.\;
     \For{$0\leq t<s$,}{\uIf{$t=0$}{
      compute $k_{0}=\underset{\underset{j\not\in\mathcal{N}_q}{j\,|\,n}}{\sum}\texttt{ord}_j(q)\cdot \nu_j^{(0)},$ where $0\leq \nu_j^{(0)}\leq   \beta(j;q)$\;
    }
    \Else{\While{$0<t<\left\lceil\frac{s}{2}\right\rceil$}{For a fixed  $\nu_j^{(t-1)}$ in $k_{t-1}$, compute $k_{t}=\underset{\underset{j\not\in\mathcal{N}_q}{j\,|\,n}}{\sum}\texttt{ord}_j(q)\cdot\nu_j^{(t)},$ where $0\leq \nu_j^{(t)}\leq
\beta(j;q)-\nu_j^{(t-1)},$\;\uIf{$2k_0+k_1+\cdots+k_t\leq n$}{
consider $k_t$,}\Else{reject $k_t$} }
    \While{$ t\geq \left\lceil\frac{s}{2}\right\rceil$}{For a fixed  $\nu_j^{(t-1)}$ in $k_{t-1}$, compute $k_{t}=\underset{\underset{i\in\mathcal{N}_q}{i\,|\,n}}{\sum}\texttt{ord}_i(q)\cdot \mu_i^{(t)}+\underset{\underset{j\not\in\mathcal{N}_q}{j\,|\,n}}{\sum}\texttt{ord}_j(q)\cdot
    \nu_j^{(t)},$ where $0\leq \mu_i^{(t)}\leq \gamma(i; q)$ and $0\leq \nu_j^{(t)}\leq  2\cdot(\beta(j;q)-\nu_j^{(t-1)})$.\;\uIf{$2k_0+k_1+\cdots+k_t\leq n$}{ consider $k_t$,}\Else{reject
$k_t$} }}
    }
 \Return{The possible parameters $(k_0, k_1, \cdots, k_{s-1})$.}\;
    }
 \caption{{Parameters of the Euclidean hull of a cyclic
serial code of length $n$ over $R$. }} \label{algo:duplicate2}
\end{algorithm}

\begin{Example} All possible parameters of Euclidean hulls of cyclic codes of
        length $n=11$ over $\mathbb{Z}_{27}$ are determined as follows.
        \begin{enumerate}
            \item The divisors of $11$ are $1$ and $11.$
\begin{description}
    \item[a)] We have $1\in \mathcal{N}_3,$ so $\texttt{ord}_1(3)=1$ and $\gamma(1;3)=1.$
    \item[b)] We have $11\not\in \mathcal{N}_3,$ so $\texttt{ord}_{11}(3)=5$ and $\beta(11;3)=1.$
\end{description}
            \item It follows that
\begin{eqnarray*}
  k_0 &=& 5\nu_{11}^{(0)}, \text{ where } 0\leq\nu_{11}^{(0)}\leq 1 \\
  k_1 &=& 5\nu_{11}^{(1)}, \text{ where } 0\leq\nu_{11}^{(1)}\leq 1-\nu_{11}^{(0)}\\
  k_2 &=& \mu_1^{(2)}+5\nu_{11}^{(2)} \text{ where } 0\leq \mu_1^{(2)}\leq 1 \text{ and } 0\leq\nu_{11}^{(2)}\leq  2(1-\nu_{11}^{(1)}).
\end{eqnarray*}
Hence, the all possible parameters $(k_0,k_1,k_2)$ of the
Euclidean hulls of cyclic codes of length $7$ over
$\mathbb{Z}_{8}$ are given in the following table
\begin{center}
\begin{tabular}{lll}
  \hline
      $k_0$ & $k_1$ & $k_2$ \\
    \hline
       0 & 0  & 0, 1, 5, 6, 10, 11 \\
         & 5  & 0, 1  \\
       5 & 0  & 0, 1   \\
  \hline
\end{tabular}
\end{center}
\end{enumerate}
\end{Example}

    \begin{Example} All the possible parameters $(k_0, k_1, k_2)$ of the Euclidean hull of a cyclic code of length $7$ over $\mathbb{Z}_8$
are determined as follows.
\begin{enumerate}
\item The divisors of $7$ are $1$ and $7.$
\begin{description}
    \item[a)] We have $1\in \mathcal{N}_2,$ so  $\texttt{ord}_1(2)=1$ and $\gamma(1;2)=1.$
    \item[b)] We have $7\not\in \mathcal{N}_2,$ so  $\texttt{ord}_7(2)=3$ and $\beta(7;2)=1.$
\end{description}
\item It follows that
\begin{eqnarray*}
  k_0 &=& 3\nu_7^{(0)}, \text{ where } 0\leq\nu_7^{(0)}\leq 1 \\
  k_1 &=& 3\nu_7^{(1)}, \text{ where } 0\leq\nu_7^{(1)}\leq 1-\nu_7^{(0)}\\
  k_2 &=& \mu_1^{(2)}+3\nu_7^{(2)} \text{ where } 0\leq \mu_1^{(2)}\leq 1 \text{ and } 0\leq\nu_7^{(2)}\leq  2(1-\nu_7^{(1)}).
\end{eqnarray*}
Hence, the all possible parameters $(k_0,k_1,k_2)$ of the
Euclidean hulls of cyclic codes of length $7$ over
$\mathbb{Z}_{8}$ are given in the following table
\begin{center}
\begin{tabular}{lll}
  \hline
      $k_0$ & $k_1$ & $k_2$ \\
    \hline
       0 & 0  & 0, 1, 3, 4, 6, 7 \\
         & 3  & 0, 1  \\
       3 & 0  & 0, 1   \\
  \hline
\end{tabular}
\end{center}
\end{enumerate}
    \end{Example}

    \begin{Example}The parameters of the Euclidean hulls of cyclic codes of length $21$ over $\mathbb{Z}_8$ are given by
        \begin{enumerate}
            \item The divisors of $21$ are $\{1,3,7,21\}.$
\begin{description}
    \item[(a)] $1; 3\in \mathcal{N}_2,$ we have  $\texttt{ord}_1(2)=1,  \texttt{ord}_3(2)=2$ and $\gamma(1;2)=\gamma(3;2)=1.$
    \item[(b)] $7; 21\not\in \mathcal{N}_2,$ we have  $\texttt{ord}_7(2)=3,\texttt{ord}_{21}(2)=6$ and $\beta(7;2)=\beta(21;2)=1.$
\end{description}
\item It follows that
\begin{eqnarray*}
  k_0 &=& 3\nu_{7}^{(0)}+6\nu_{21}^{(0)}, \text{ with } 0\leq\nu_j^{(0)}\leq 1, \text{ where } j\in\{7, 21\}. \\
  k_1 &=& 3\nu_7^{(1)}+6\nu_{21}^{(1)}, \text{ with } 0\leq\nu_j^{(1)}\leq 1-\nu_j^{(0)}, \text{ where } j\in\{7, 21\}.\\
  k_2 &=& \mu_1^{(2)}+2\mu_3^{(2)}+3\nu_7^{(2)}+6\nu_{21}^{(2)}, \text{ with } 0\leq \mu_i^{(2)}\leq 1 \text{ and } 0\leq\nu_j^{(2)}\leq  2(1-\nu_j^{(1)}), \text{ where } i\in\{1, 3\}, \text{ and } j\in\{7, 21\}.
\end{eqnarray*}
Hence, the all possible parameters $(k_0,k_1,k_2)$ of the
Euclidean hulls of cyclic codes of length $21$ over
$\mathbb{Z}_{8}$ are given in the following table
\begin{center}
\begin{tabular}{lll}
  \hline
      $k_0$ & $k_1$ & $k_2$ \\
    \hline
       0 & 0  & 0, 1, 2, 3, $\cdots,$ 21 \\
         & 3  & 0, 1, 2, 3, 6, 7, 8, 9, 12, 13, 14, 15  \\
         & 6  & 0, 1, 2, 3, 4, 5, 6, 7, 8, 9   \\
         & 9  & 0, 1, 2, 3   \\
      3  & 0  & 0, 1, 3, $\cdots$, 15  \\
         & 6  & 0, 1, 2, 3, 4, 5, 6, 7, 8, 9  \\
      6  & 0  & 0, 1, 3, $\cdots$, 9    \\
         & 3  & 0, 1, 2, 3, 6, 7, 8, 9, 12, 13, 14, 15  \\
      9  & 0  & 0, 1, 2, 3   \\
  \hline
\end{tabular}
\end{center}
\end{enumerate}
    \end{Example}

\begin{Corollary}\label{c*} The set $\aleph(n,s,q)$ of $q$-dimensions of the Euclidean hull of a cyclic
serial code of length $n$ over $R,$ is given by
$$\aleph(n,s,q)=\left\{\underset{\underset{i\in\mathcal{N}_q}{i\,|\,n}}{\sum}\texttt{ord}_i(q)\left(\sum\limits_{l=1}^{\gamma(i;q)}
\vartriangle_{il}\right)+\underset{\underset{j\not\in\mathcal{N}_q}{j\,|\,n}}{\sum}\texttt{ord}_j(q)\left(\sum\limits_{h=1}^{\beta(j;q)}\blacktriangle_{jh}\right)
\mid    %
\begin{array}{ll}
   0\leq \vartriangle_{il}\leq s-\left\lceil\frac{s}{2}\right\rceil  \\
0\leq \blacktriangle_{jh}\leq s
\end{array}   \right\}.$$
\end{Corollary}

\begin{Proof} Let $C$ be a cyclic serial code of length $n$ over
$R$ with triple-sequence
$$(\textbf{x}^{\circ},\textbf{y}^{\bullet},\textbf{z}^{\bullet})=\left(
(((x^{(a)}_{il})_{0\leq a< s})^{\circ}), (((y^{(a)}_{jh})_{0\leq
a< s})^{\bullet}), (((z^{(a)}_{jh})_{0\leq a<
s})^{\bullet})\right)$$ in $\mathcal{E}_n(q,s).$ From Theorem
\ref{thm2}, the parameters $(k_0, k_1, \cdots, k_{s-1})$ of
$\mathcal{H}_0(C)$ where for all $0\leq t\leq s-1$,
$$
k_t
=\underset{\underset{i\in\mathcal{N}_q}{i\,|\,n}}{\sum}\texttt{ord}_i(q)\cdot\left(\sum\limits_{i=1}^{\gamma(i;q)}(u_{il}^{[t-1]}-u_{il}^{[t]})\right)+\underset{\underset{j\not\in\mathcal{N}_q}{j\,|\,n}}{\sum}\texttt{ord}_j(q)\cdot\left(\sum\limits_{h=1}^{\beta(j;q)}(\varepsilon_{jh}^{(t-1)}-\varepsilon_{jh}^{(t)})\right).
$$
Thus the $q$-dimension of $\mathcal{H}_0(C)$ is
$\sum\limits_{t=0}^{s-1}(s-t)k_t.$ It follows that
$$\texttt{dim}_q(C)=\underset{\underset{i\in\mathcal{N}_q}{i\,|\,n}}{\sum}\texttt{ord}_i(q)\cdot\left(\sum\limits_{i=1}^{\gamma(i;q)}\vartriangle_{il}\right)+\underset{\underset{j\not\in\mathcal{N}_q}{j\,|\,n}}{\sum}\texttt{ord}_j(q)\cdot\left(\sum\limits_{h=1}^{\beta(j;q)}\blacktriangle_{jh}\right).$$
From Remark\,\ref{r00},
\begin{equation*}
  \vartriangle_{il} = \sum\limits_{t=0}^{s-1}\vartriangle^{(t)}_{il}=\sum\limits_{t=\left\lceil\frac{s}{2}\right\rceil}^{s-1}\vartriangle^{(t)}_{il}\; \leq  s-\left\lceil\frac{s}{2}\right\rceil,
\end{equation*}
and if $j\in\mathcal{N}_q$ then $\blacktriangle_{j}=0.$ Otherwise,
\begin{equation*}
  \blacktriangle_{jh}=\sum\limits_{t=0}^{s-1}\blacktriangle^{(t)}_{jh} \, =  \sum\limits_{t=0}^{\left\lceil\frac{s}{2}\right\rceil-1}\blacktriangle^{(t)}_{jh}+\sum\limits_{t=\left\lceil\frac{s}{2}\right\rceil}^{s-1}\blacktriangle^{(t)}_{jh}\,\leq  \underset{0\leq b \leq s-\left\lceil\frac{s}{2}\right\rceil}{\max }\left\{\left(\left\lceil\frac{s}{2}\right\rceil+b\right)+2\left(s-\left\lceil\frac{s}{2}\right\rceil-b\right )\right\}=s.
\end{equation*}
\end{Proof}

\section{The average $q$-dimension}

We will denote by $\mathcal{C}(n; R)$  the set of all cyclic
serial codes over length $n$ over $R$. The average $q$-dimension
of the Euclidean hull of cyclic of length $n$ over $R$ is
$$\texttt{E}_R(n)=\sum\limits_{C\in\mathcal{C}(n; R)}\frac{\texttt{dim}_q(\mathcal{H}_0(C))}{|\mathcal{C}(n; R)|}.$$
In this section, an explicit formula for $\texttt{E}_R(n)$ and
bounds are given in terms of $\texttt{B}_{n,q}$ where
$$\texttt{B}_{n,q}=\texttt{deg}\underset{\underset{i\in\mathcal{N}_q}{i\,|\,n}}{\prod}\left(\prod\limits_{l=1}^{\gamma(i;q)}\Omega(G_{il})\right)=\underset{\underset{i\in\mathcal{N}_q}{i\,|\,n}}{\sum}\phi(i).$$

Consider the maps
\begin{align}\label{var1}
\begin{array}{cccc}
  \vartriangle : & \mathcal{E}_s & \rightarrow & \mathbb{N} \\
    & (x^{(0)}, \cdots, x^{(s-1)}) & \mapsto &
    \sum\limits_{t=0}^{s-1}\min\left\{\sum\limits_{a=0}^{t}x^{(a)};1-\sum\limits_{a=0}^{s-t-1}x^{(a)}\right\},
\end{array}
\end{align}
and $\blacktriangle:
\mathcal{E}_s\times\mathcal{E}_s\rightarrow\mathbb{N}$ defined as
\begin{align}\label{var2}\blacktriangle(\textbf{y},\textbf{z})=\sum\limits_{t=0}^{s-1}\left(\min\left\{\sum\limits_{a=0}^{t}y^{(a)};1-\sum\limits_{a=0}^{s-t-1}z^{(a)}\right\}+
\min\left\{\sum\limits_{a=0}^{t}z^{(a)};1-\sum\limits_{a=0}^{s-t-1}y^{(a)}\right\}\right),\end{align}
where $(\textbf{y},\textbf{z})=((y^{(0)}, \cdots, y^{(s-1)}),
(z^{(0)}, \cdots, z^{(s-1)})).$

Let $\tau\in\aleph(n,s,q)$ an element in the set defined in
Corollary \ref{c*}. Then $\tau$ is the $q$-dimension of the
Euclidean hull of a cyclic serial code of length $n$ over $R$. The
following result gives the number of cyclic serial codes of length
$n$ over $R$ whose Euclidean hulls have $q$-dimension $\tau$. 

\begin{Proposition} Let $n$ be a positive integer such that $\texttt{gcd}(n,
p)=1$ and $\tau\in\aleph(n,s,q)$ where $\aleph(n,s,q)$ is
described in Corollary \ref{c*}. The number $\wp(n, \tau; R)$ of
cyclic serial codes of length $n$ over $R$ whose Euclidean hulls
have $q$-dimension $\tau$ is given by:
$$
\wp(n, \tau; R)=\sum\limits_{(((\vartriangle_{il})^\circ),
((\blacktriangle_{jh})^\bullet))\in
\Upsilon(\tau)}\left(\underset{\underset{i\in\mathcal{N}_q}{i\,|\,n}}{\prod}\prod\limits_{l=1}^{\gamma(i;q)}\psi_s(\vartriangle_{il})\right)\left(\underset{\underset{j\not\in\mathcal{N}_q}{j\,|\,n}}{\prod}\prod\limits_{h=1}^{\beta(j;q)}\rho_s(\blacktriangle_{jh})\right),
$$
where
$$\psi_s(\vartriangle_{il})=|\{\textbf{x}\in\mathcal{E}_s\;:\;\vartriangle(\textbf{x})=\vartriangle_{il}\}|,\;
\rho_s(\blacktriangle_{jh})=|\{(\textbf{y},\textbf{z})\in\mathcal{E}_s\times\mathcal{E}_s\;:\;\blacktriangle(\textbf{y},\textbf{z})=\blacktriangle_{jh}\}|,$$
and
$$\Upsilon(\tau)=\left\{(((\vartriangle_{il})^\circ),
((\blacktriangle_{jh})^\bullet))\,:\,
\underset{\underset{i\in\mathcal{N}_q}{i\,|\,n}}{\sum}\texttt{ord}_i(q)\left(\sum\limits_{l=1}^{\gamma(i;q)}
\vartriangle_{il}\right)+\underset{\underset{j\not\in\mathcal{N}_q}{j\,|\,n}}{\sum}\texttt{ord}_j(q)\left(\sum\limits_{h=1}^{\beta(j;q)}\blacktriangle_{jh}\right)
=\tau\right\}.$$
\end{Proposition}

The above expresion of
$\texttt{E}_R(n)=\sum\limits_{\tau\in\aleph(n,s,q)}\frac{\tau\cdot\wp(n,
\tau; R)}{|C\in\mathcal{C}(n; R)|},$  might lead to a tedious and
lengthy computation. The remainder of the section will show an
alternative simpler expresion for the expected value.

\begin{Lemma}\label{l*} Consider the  random variable
    $\vartriangle$ defined in (\ref{var1})  with uniform probability. The expected value $\texttt{E}(\vartriangle)$  is given by:
$$\texttt{E}(\vartriangle)=\frac{\left\lceil\frac{s}{2}\right\rceil\left(s-\left\lceil\frac{s}{2}\right\rceil\right)}{s+1}=\left\{%
\begin{array}{ll}
    \frac{s^2}{4(s+1)}, & \hbox{if $s$ even;} \\
    \frac{s-1}{4}, & \hbox{if $s$ odd.}
\end{array}%
\right.    $$
\end{Lemma}

\begin{Proof} Let $t\in\{0;1;\cdots; s-1\}$ and $\textbf{x}=(x^{(0)}, \cdots,
x^{(s-1)})\in\mathcal{E}_s$. Set
$$\vartriangle^{(t)}_{(\textbf{x})}=\min\left\{
\sum\limits_{a=0}^{t}x^{(a)};
1-\sum\limits_{a=0}^{s-t-1}x^{(a)}\right\}\in\{0; 1\}.$$ Then
$\vartriangle^{(t)}_{(\textbf{x})} =1$ if and only if  $2t\geq s
\text{ and }, \sum\limits_{a=s-t}^{t}x^{(a)}_{il}=1$. Thus for all
$\eta\in\mathbb{N},$ we have
$|\{\textbf{x}\in\mathcal{E}_s\;:\;\vartriangle^{(t)}_{(\textbf{x})}
=\eta\}|=\left\{%
\begin{array}{ll}
    2t-s+1, & \hbox{if $t\geq\left\lceil\frac{s}{2}\right\rceil$ and $\eta=1$;} \\
    0, & \hbox{otherwise.}
\end{array}%
\right.    $

Therefore,
\begin{eqnarray*}
  |\{\textbf{x}\in\mathcal{E}_s\;:\;\vartriangle(\textbf{x})=\eta\}| &=& \left\{%
\begin{array}{ll}
   \sum\limits_{t=\left\lceil\frac{s}{2}\right\rceil}^{s-1}(2t-s+1), & \hbox{ if $\eta=s-\left\lceil\frac{s}{2}\right\rceil$; } \\
    0, & \hbox{ otherwise. }
\end{array}%
\right. \\
    &=& \left\{%
\begin{array}{ll}
   \left\lceil\frac{s}{2}\right\rceil\left(s-\left\lceil\frac{s}{2}\right\rceil\right), & \hbox{ if $\eta=s-\left\lceil\frac{s}{2}\right\rceil$; } \\
    0, & \hbox{ otherwise; }
\end{array}%
\right.
\end{eqnarray*}

Since $|\mathcal{E}_s|=s+1$ and
$\texttt{P}(\{\textbf{x}\in\mathcal{E}_s\;:\;\vartriangle(\textbf{x})=\eta\})=\frac{|\{\vartriangle(\textbf{x})=\eta\}|}{|\mathcal{E}_s|},$
it follows that, $$ \texttt{E}(\vartriangle) =
\sum\limits_{\eta\in\mathbb{N}}\eta\texttt{P}(\{\textbf{x}\in\mathcal{E}_s\;:\;\vartriangle(\textbf{x})=\eta\})=
\frac{\left\lceil\frac{s}{2}\right\rceil\left(s-\left\lceil\frac{s}{2}\right\rceil\right)}{s+1}.$$
\end{Proof}

\begin{Lemma}\label{l**} Consider the random variable
    $\blacktriangle: \mathcal{E}_s\times\mathcal{E}_s\rightarrow\mathbb{N}$ defined in
    (\ref{var2}) with uniform distribution. The expected value $\texttt{E}(\blacktriangle)$ is
 is given by
$$\texttt{E}(\blacktriangle)= \frac{s(2s+1)}{3(s+1)}.$$
\end{Lemma}

\begin{Proof} From Corollary \ref{c*}, for any $(\textbf{y}, \textbf{z})\in \mathcal{E}_s\times\mathcal{E}_s,$ $0\leq \blacktriangle(\textbf{y}, \textbf{z})\leq s.$
Let $$\mathcal{E}_s(\eta)=\{(\textbf{y}, \textbf{z})\in
\mathcal{E}_s\times\mathcal{E}_s\;:\; \blacktriangle(\textbf{y},
\textbf{z})=\eta\},$$ for $0\leq \eta\leq s$. Now, $$
|\mathcal{E}_s(\eta)|=\left\{%
\begin{array}{ll}
    2(\eta+1), & \hbox{if $0\leq\eta\leq s-1$;} \\
    s+1, & \hbox{if $\eta=s$.}
\end{array}%
\right.
$$
 Thus
\begin{eqnarray*}
\texttt{E}(\blacktriangle) &=& \frac{1}{(s+1)^2}\sum\limits_{\eta=0}^{ s}\eta|\mathcal{E}_s(\eta)|;  \\
    &=& \frac{1}{(s+1)^2}\left(\sum\limits_{\eta=1}^{s-1}2\eta(\eta+1)+s(s+1)\right); \\
    &=& \frac{s(2s^2+3s+1)}{3(s+1)^2}.
\end{eqnarray*}
\end{Proof}

\begin{Theorem} The average $q$-dimension of the Euclidean hull of cyclic
serial codes from $\mathcal{C}(n; R)$ is
$$\texttt{E}_R(n)=\left\{%
\begin{array}{ll}
    \left(\frac{(2s+1)s}{6(s+1)}\right)n-\left(\frac{\left(s+2\right)s}{12\left(s+1\right)}\right)\texttt{B}_{n,q}, & \hbox{if $s$ even;} \\
    \left(\frac{(2s+1)s}{6(s+1)}\right) n-\left(\frac{s^{2}+2s+3}{12\left( s+1\right)}\right)\texttt{B}_{n,q}, & \hbox{if $s$ odd.}
\end{array}%
\right.$$ where
$\texttt{B}_{n,q}=\underset{\underset{i\in\mathcal{N}_q}{i\,|\,n}}{\sum}\phi(i).$
\end{Theorem}

\begin{Proof} Let $Y$ be the random variable that takes as value
$\texttt{dim}_q(\mathcal{H}_0(C))$ when we choose at random a
cyclic serial code from $\mathcal{C}(n; R)$ with uniform
probability. Then $\texttt{E}(Y)=\texttt{E}_R(n).$ By
Lemma\,\ref{lcs}, there exists an one-to-one correspondence
between $\mathcal{C}(n; R)$, and $\mathcal{E}_n(q,s)$. Therefore,
choosing a cyclic serial code $C$ from $\mathcal{C}(n,R)$ their
probabilities are identical. By Corollary \ref{c*}, we obtain
$$Y=\underset{\underset{i\in\mathcal{N}_q}{i\,|\,n}}{\sum}\texttt{ord}_i(q)\left(\sum\limits_{l=1}^{\gamma(i;q)}
\vartriangle_{il}\right)+\underset{\underset{j\not\in\mathcal{N}_q}{j\,|\,n}}{\sum}\texttt{ord}_j(q)\left(\sum\limits_{h=1}^{\beta(j;q)}\blacktriangle_{jh}\right).$$
For all $i$ and $j$ dividing $n$ such that $i\in\mathcal{N}_q$ and
$j\not\in\mathcal{N}_q,$ from Lemmas \ref{l*} and \ref{l**}, we
note that $\texttt{E}(\vartriangle_{il})=\texttt{E}(\vartriangle)$
and $\texttt{E}(\blacktriangle_{jh})=\texttt{E}(\blacktriangle).$
So, we get
\begin{eqnarray*}
  \texttt{E}(Y) &=& \underset{\underset{i\in\mathcal{N}_q}{i\,|\,n}}{\sum}\texttt{ord}_i(q)\left(\sum\limits_{l=1}^{\gamma(i;q)}\texttt{E}(\vartriangle)\right)+\underset{\underset{j\not\in\mathcal{N}_q}{j\,|\,n}}{\sum}\texttt{ord}_j(q)\left(\sum\limits_{h=1}^{\beta(j;q)}\texttt{E}(\blacktriangle)\right); \\
    &=& \underset{\underset{i\in\mathcal{N}_q}{i\,|\,n}}{\sum}\phi(i)\texttt{E}(\vartriangle_{il})+\underset{\underset{j\not\in\mathcal{N}_q}{j\,|\,n}}{\sum}\frac{\phi(j)}{2}\texttt{E}(\blacktriangle_{jh});\\
    &=& \texttt{B}_{n,q}\texttt{E}(\vartriangle)+\left(\frac{n-\texttt{B}_{n,q}}{2}\right)\texttt{E}(\blacktriangle);\\
    &=& \frac{n}{2}\texttt{E}(\blacktriangle)-\texttt{B}_{n,q}\cdot\left(\frac{1}{2}\texttt{E}(\blacktriangle)-\texttt{E}(\vartriangle)\right).
   \end{eqnarray*}
From Lemmas \ref{l*} and \ref{l**}, we have
$$\texttt{E}_R(n)=\left\{%
\begin{array}{ll}
    \left(\frac{(2s+1)s}{6(s+1)}\right)n-\left(\frac{\left(s+2\right)s}{12\left(s+1\right)}\right)\texttt{B}_{n,q}, & \hbox{if $s$ even;} \\
    \left(\frac{(2s+1)s}{6(s+1)}\right) n-\left(\frac{s^{2}+2s+3}{12\left( s+1\right)}\right)\texttt{B}_{n,q}, & \hbox{if $s$ odd.}
\end{array}%
\right.$$
\end{Proof}

From \cite{Ske03}, we have $\texttt{B}_{n,q}=n$ if
$n\in\mathcal{N}_q$ and $1\leq \texttt{B}_{n,q}\leq \frac{2n}{3}$
if $n\not\in\mathcal{N}_q.$ Thus
 \begin{itemize}
    \item If $n\in\mathcal{N}_q$, then $$\texttt{E}_R(n)=\left\{%
\begin{array}{ll}
    \frac{s^2n}{4(s+1)}, & \hbox{if $s$ even;} \\
    \frac{n(s-1)}{4}, & \hbox{if $s$ odd.}
\end{array}%
\right.$$
    \item If $n\not\in\mathcal{N}_q$, then $$\left\{%
\begin{array}{ll}
   \frac{(5s+1)sn}{18(s+1)}\leq \texttt{E}_R(n) \leq\frac{2n(2s+1)s-(s+2)s}{12(s+1)}, & \hbox{if $s$ even;} \\
    \frac{(5s^{2}+s-3)n}{18(s+1)}\leq \texttt{E}_R(n) \leq\frac{2ns(2s+1)-(s^2+2s+3)}{12(s+1)}, & \hbox{if $s$ odd.}
\end{array}%
\right.$$
 \end{itemize}

\begin{Remark} Note that 
    $\texttt{E}_R(n)$ grows at the
    same rate as $ns$ when $s$ and $n$ tend to
    infinity. Thus, the upper limit of the sequence
    $\left(\frac{\texttt{E}_R(n)}{sn}\right)_{\underset{\texttt{gcd}(p,n)=1}{(s,n)\in(\mathbb{N}\backslash\{0\})^2}}$
    is at most $\frac{1}{3}$ and its its lower limit is at least
    $\frac{5}{18}.$
\end{Remark}

\section{Conclusion}

The Galois hulls of cyclic serial codes of length $n$ over an
arbitrary finite chain ring with parameters $(p,r,a,e,s)$ have
been investigated. Especially, the parameters and the average of
the $q$-dimension of the Euclidean hull of cyclic codes are
studied in terms of triple-sequences.  The parameters and the
average $p^r$-dimensions of the Euclidean hulls of cyclic serial
codes of arbitrary length have been determined as well.
Asymptotically, it has been shown that the average of
$p^r$-dimension of the Euclidean hull of cyclic serial codes of
length over $R$ grows the same rate as the length of the codes. An
extension of this paper to the case of the hulls of cyclic or
constacyclic codes over finite chain rings is an interesting
research problem as well. It would be interesting to study the
properties of Euclidean hulls of negacyclic serial codes.

\bibliographystyle{elsarticle-num}

\end{document}